%% file: qpip.ar2.tex
\documentclass[a4paper,twoside,english,10pt]{article}
\usepackage[sumlimits]{amsmath}
\usepackage{times}
\usepackage{type1cm}

\newcommand{\vb}{\vspace{3mm}\noindent}

\usepackage[margin=2.45cm]{geometry}
\usepackage{babel}

 \usepackage[pdfauthor={Elad Eban, Dorit Aharonov,Michael
   Ben-Or},pdftitle={Interactive Proofs For Quantum Computation},
 pdfcreator={Elad Eban},colorlinks=false]{hyperref}

\input{def}

\newcommand{\dnote}[1]{\textcolor{red}{\small {\textbf{(Dorit:}
#1\textbf{) }}}}
\newcommand{\qw}{\mbox{\emph{Q-Wave}}}
\newcommand{\qc}{\mbox{\textsf{Q-CIRCUIT}}}
\newcommand{\QPIP}{\textsf{QPIP}}
\newcommand{\QAS}{\textsf{QAS}}

\newcommand{\wt}[1]{{\widetilde{#1}}}
\newcommand{\odots}{{\otimes{\ldots}\otimes}}

\newenvironment{proof}{\noindent\textit{Proof: }}{$\Box $}

\begin{document}

\author{Dorit Aharonov\footnote{School of Computer Science, The Hebrew
    University of Jerusalem, Israel. $\{$doria,benor,elade$\}$@cs.huji.ac.il}
  \and Michael Ben-Or$^*$ \and Elad Eban$^*$}

\title{Interactive Proofs For Quantum Computations}

\maketitle
\thispagestyle{empty}
\begin{abstract}
The widely held belief that \BQP\ strictly contains \BPP\ raises fundamental
questions: Upcoming generations of quantum computers might already be too
large to be simulated classically. Is it possible to experimentally test that
these systems perform as they should, if we cannot efficiently compute
predictions for their behavior? Vazirani has asked \cite{vazirani07}:
If computing predictions
for Quantum Mechanics requires exponential resources,
is Quantum Mechanics a falsifiable theory? In
cryptographic settings, an untrusted future company wants to sell a quantum
computer or perform a delegated quantum computation. Can the customer be
convinced of correctness without the ability to compare results to
predictions?

To provide answers to these questions, we define Quantum Prover Interactive
Proofs (\QPIP). Whereas in standard Interactive Proofs \cite{goldwasser1985kci} the prover is
computationally unbounded, here our
prover is in \BQP, representing a quantum computer.
The verifier models our current computational capabilities: it is a \BPP\
machine, with access to few qubits. Our main theorem can be roughly stated
as: "Any language in \BQP\ has a \QPIP, and moreover,
a fault tolerant one''. We provide two proofs. The
simpler one uses a new (possibly of independent interest) quantum
authentication scheme (\QAS) based on random Clifford elements. This \QPIP\
however, is not fault tolerant. Our second protocol uses polynomial codes \QAS\
due to Ben-Or, Cr{\'e}peau, Gottesman, Hassidim, and Smith \cite{benor2006smq},
combined with quantum fault tolerance and secure multiparty quantum
computation techniques. A slight
modification of our constructions makes the protocol ``blind'': the quantum
computation and input remain unknown to the prover.

After we have derived the results, we have learnt that
Broadbent, Fitzsimons, and Kashefi \cite{broadbent2008ubq} have independently
derived "universal blind quantum computation'' using completely different
methods (measurement based
quantum computation). Their construction implicitly implies
similar implications.
\end{abstract}
\newpage
\setcounter{page}{1}
\section{Introduction}
\subsection{Motivation}\label{sec:funda}
As far as we know today, the quantum mechanical description
of many-particle systems requires exponential resources to simulate.
This has the following fundamental implication:
the results of an experiment conducted on a many-particle physical system
described by quantum mechanics, cannot be predicted (in general)
by classical computational devices, in any reasonable amount of time.
This important realization (or belief),
which stands at the heart of the interest in quantum computation,
led Vazirani to ask \cite{vazirani07}: Is quantum mechanics a falsifiable
physical theory? Assuming that small quantum systems obey quantum
mechanics to an extremely
high accuracy, it is still possible that the physical description of
large systems deviates significantly from quantum mechanics.
Since there is no efficient way to make the predictions of the experimental
outcomes for most large quantum systems, there is no way
to test or falsify this possibility experimentally, using the usual
scientific paradigm, as described by Popper.

This question has practical implications.
Experimentalists who attempt to realize quantum computers
would like to know how to test that their systems indeed perform
the way they should. But most tests cannot be compared to any
predictions! The tests whose predictions can in fact be computed,
do not actually test the more interesting aspects of quantum mechanics, namely
those which cannot be simulated efficiently classically.

The problem arises in cryptographic situations as well.
Consider for example, a company called \qw\ which
is trying to convince a certain potential customer that the system it had
managed to build is in fact a quantum computer of $100$ qubits.
How can the customer, who cannot make predictions of the outcomes
of the computations made by the machine, test that the machine is indeed
a quantum computer which does what it is claimed to do?
Given the amounts of grant money and prestige involved,
the possibility of
dishonesty of experimentalists
and experimentalists' bias inside the academia should not be
ignored either \cite{roodman2003bap, BlindWiki}.

There is another related question that stems from cryptography.
It is natural to expect that the first
generations of quantum computers will be extremely expensive,
and thus quantum
computations would be delegated to untrusted
companies.  Is there any way for the costumer to trust the
outcome, without the need to trust the company which performed the
computation, even though the costumer cannot verify the
outcome of the computation (since he cannot simulate it)?
And even if the company is honest,
can the costumer detect innocent errors in such a computation?

Vazirani points out \cite{vazirani07} that in fact, an answer
to these questions is already given in the form of Shor's algorithm.
Indeed, quantum mechanics does not seem to be falsifiable
using the {\it usual} scientific paradigm, assuming that \BQP\ is
strictly lager than \BPP. However, Shor's algorithm does provide a way
for falsification, by means of
an experiment which lies outside of the scientific paradigm:
though its result cannot be {\it predicted} and then compared to the
experimental outcome, it can be {\it verified} once the outcome
of the experiment is known (by simply taking the product of the factors
and checking that this gives the input integer).

This, however, does not fully address the issues raised above.
Let us take for example the case of the company trying to convince a
costumer that the system it is trying to sell is indeed a quantum computer
of $100$ qubits. Such a system is already too big to simulate classically;
However, any factoring algorithm that is run on a system of a $100$ qubits
can be easily performed by today's classical technology.
For delegated quantum computations, how can Shor's algorithm help
in convincing a costumer of correctness of, say, the computation of the \BQP\
complete problem of approximating the Jones polynomial
\cite{aharonov2006pqa,jonesHardness}?
As for experimental results, it is difficult
to rigorously state what is exactly falsified or verified
by the possibility to apply Shor's algorithm.
Finally, from a fundamental point of view, there is a fundamental difference
between being convinced of the ability to factor, and testing
universal quantum evolution.

We thus pose the following main question:
Can one be convinced of the correctness of the computation of {\it any}
polynomial quantum circuit? Does a similar statement to the one above,
regarding Shor's algorithm, apply for universal quantum computation?
Alternatively, can one be convinced of
the ``correctness'' of the quantum mechanical description of
any quantum experiment that can be conducted in the laboratory,
even though one cannot compute any predictions for the outcomes of this
experiment? In this paper we address the above fundamental question
in a rigorous way. We do this by taking a computational point of
view on the interaction between
the supposed quantum computer, and the entity which attempts to
verify that it indeed computes what it should.

\subsection{Quantum Prover Interactive Proofs (\QPIP)}
Interactive proof systems, defined by Goldwasser, Micali and Rackoff
\cite{goldwasser1985kci}, play a crucial role in the theory of computer
science.  Roughly, a language $\mcL$ is said to have an interactive proof if
there exists a computationally unbounded prover (denoted $\mcP$) and a \BPP\
verifier ($\mcV$) such that for any $x\in\mcL$, $\mcP$ convinces $\mcV$ of the
fact that $x\in\mcL$ with probability $\ge \frac 2 3$ (completeness).
Otherwise,
when $x\notin\mcL$ any prover fails to convince $\mcV$ with probability higher
than $\frac 1 3$ (soundness).

Shor's factoring algorithm \cite{shor1997pta} can be viewed as an interactive
proof of a very different kind: one between a classical \BPP\ verifier, and a
quantum \textit{polynomial time} (\BQP) prover, in which
the prover convinces the verifier of the factors of a given number
(this can be easily converted to the usual \IP\ formalism of membership
in a language).
This is a quantum interactive proof of
a very different kind than quantum interactive proofs previously
studied in the literature \cite{watrous2003phc}, in which the prover is
an \emph{unbounded}
quantum computer, and the \emph{verifier} is a \BQP\ machine.

Clearly, such an interactive proof between a \BQP\ prover and
a \BPP\ verifier exists for any problem inside \NP $\cap$ \BQP.
However, it is widely believed that \BQP\ is not contained in \NP\ (
and in fact not even in the polynomial hierarchy). The main idea
of the paper is to generalize the above interactive point of view
of Shor's's algorithm, and show that with this generalization,
a verifier can be convinced of the result of {\it any} polynomial
quantum circuit, using interaction with the prover - the quantum
computer.

To this end we define a new model of quantum interactive proofs which we call
quantum prover interactive proofs (\QPIP). The simplest definition would be
an interactive proof in which the prover is a \BQP\ machine and the verifier a
\BPP\ classical machine.  In some sense, this model captures the
possible interaction
between the quantum world (for instance, quantum systems in the lab) and the
classical world. However, this model does not suffice for
our purposes.
We therefore modify it a little, and
allow the verifier additional access to a constant number of qubits.
The verifier can be viewed as modeling our current computational
abilities, and so in some sense, the verifier in the following system
represents ``us''.

\begin{deff}\label{def:QPIP} Quantum Prover Interactive Proof (\QPIP)
is an interactive proof system with the following properties:
\begin{itemize}
\item The prover is computationally restricted to \BQP.
\item The verifier is a hybrid quantum-classical machine. Its classical part is
      a \BPP\ machine. The quantum part is a register of $c$ qubits (for some
      constant $c$), on which the prover
 can perform arbitrary quantum operations. At
      any given time, the verifier is not allowed to possess
      more than $c$ qubits. The
      interaction between the quantum and classical parts is the usual one: the
      classical part controls which operations are to be performed on the
      quantum register, and outcomes of measurements of the quantum register can
      be used as input to the classical machine.
\item There are two communication channels: one quantum and one
      classical.
\end{itemize}
The completeness and soundness conditions are identical to the \IP\ conditions.
\end{deff}

Abusing notation, we
denote the class of languages for which such a proof
 exists also
by $\QPIP$.

\subsection{Main Results}

\begin{deff} The promise problem \qc\ consists of a quantum circuit made of a
  sequence of gates, $U=U_T{\ldots}U_1$, acting on $n$ input bits. The task is
  to distinguish between two cases:
\begin{eqnarray*}
  \qc_{\textmd{YES}}&:
  \|(\left(\ket 0 \bra 0 \otimes \mcI_{n-1}\right)U\ket{\bar{ 0}} \|^2\ge \frac 2 3\\
  \qc_{\textmd{NO}}\;\,&:
  \|(\left(\ket 0 \bra 0 \otimes \mcI_{n-1}\right)U\ket{\bar{ 0}} \|^2\le \frac 1 3
\end{eqnarray*}
\end{deff}

\qc\ is a \BQP\ complete problem, and moreover, this remains true for
other soundness and completeness parameters $0<s,c<1$, if
 $c-s>\frac 1 {Poly(n)}$.
Our main result is:
\begin{thm}\label{thm:qcircuit}
  The language \qc\ has a \QPIP.
\end{thm}

Since \qc\ is \BQP\ complete, and \QPIP\ is trivially
inside \BQP, we have:

\begin{thm} \label{thm:main}
$\BQP\ = \QPIP$.
\end{thm}

Thus, a \BQP\ the prover can convince the verifier of any language he can
compute. We remark that our definition of \QPIP\ is asymmetric -
the verifier is ``convinced'' only if the quantum circuit
outputs $1$. This asymmetry seems irrelevant in
our context of verifying correctness of quantum computations.  Indeed, it is
possible to define a symmetric version of \QPIP,
(we denote it by $\QPIP^{sym}$)
in which the verifier is convinced of {\it correctness} of the prover's outcome
(in both $0$ and $1$ cases) rather than of membership of the input in the
language, namely in the $1$ case only.  That $\BQP=\QPIP^{sym}$ follows quite
easily from the fact that $\BQP$ is closed under complement (see \pen{app:sym}).

Moreover, the above results apply in a realistic setting, namely with noise:

\begin{thm}\label{thm:ft}
\Th{thm:qcircuit} holds also when the quantum communication
and computation devices are subjected to the usual local noise model
assumed in quantum fault tolerance settings.
\end{thm}

In the works
\cite{childs2001saq,blind} a related
question was raised: in our cryptographic setting, if we distrust
the company performing the delegated quantum computation,
we might want to keep both the input and the
function which is being computed secret.
Can this be done while maintaining the
confidence in the outcome? A simple modification of our protocols gives

\begin{thm}\label{thm:blind}
\Th{thm:ft} holds also in a blind setting,
namely, the prover does not get any
 information regarding the function being computed and its input.
\end{thm}

We note that an analogous result for \NP-hard problems
was shown already in the late $80$'s to be impossible
unless the polynomial hierarchy collapses \cite{abadi1987hio}.

\subsection{Proofs Overview (and More
Results About Quantum Authentication Schemes)}
Our main tool is quantum authentication schemes (\QAS) \cite{barnum2002aqm}.
Roughly, a \QAS\ allows two parties to communicate in the following way: Alice
sends an encoded quantum state to Bob. The scheme is secure if upon decoding,
Bob gets the same state as Alice had sent unless it was altered, whereas if the
state had been altered, then Bob's chances of declaring valid a wrong state are
small.  The basic idea is that similar security can be achieved, even if the
state needs to be rotated by unitary gates, as long as the verifier can control
how the unitary gates affect the authenticated states.  Implementing this
simple
idea in the context of fault tolerance encounters several complications,
which we
explain later.

We start with a simple \QAS\ and \QPIP, which we do not know how to make fault
tolerant, but which demonstrates some of the ideas and might be of interest on
its own due to its simplicity.

\vb\textbf{Clifford \QAS\ based \QPIP}. We present a new, simple and efficient
\QAS, based on random Clifford group operations (it is reminiscent of Clifford
based quantum $t$-designs \cite{ambainis2007qtd,ambainis2008tre}).  To encode a
state of $m$ qubits, tensor the state with $d$ qubits in the state $\ket{0}$,
and apply a random Clifford operator on the $m+d$ qubits.  The security proof of
this \QAS\ uses a combination of known ideas.  We first prove that any attack of
Eve is mapped by the random Clifford operator to random Pauli operators.  We then show
that those are detected with high probability.  This \QAS\ might be interesting
on its own right due to its simplicity.

To construct a \QPIP\ using this \QAS, we simply use the prover as an untrusted
storage device: the verifier asks the prover for the authenticated qubits on
which he would like to apply the next gate, decodes them, applies the gate,
encodes them back and sends them to the prover. The proof of security is quite
straight forward given the security of the \QAS.

\ignore{
The idea of the Clifford authentication scheme is the following:
If one wants to encode an
$m$ qubit state, one can use a  random subspace of dimension $2^m$ in a
space of dimension $2^{m+d}$. To realize this, we need to
efficiently select a random subspace, which is impossible.  Instead,
we rotate a given subspace by a random Clifford operation, which
turns out to provide enough randomization.
We prove that for any encoded state and any intervention, if we average over
the random Clifford rotation, the
intervention takes the authenticated state far enough outside
the random subspace, which  will be noticeable upon decoding.
}

Due to the lack of structure of the authenticated states, we do not
know how to make the prover apply gates on the states without revealing
the key. This seems to be necessary for fault tolerance.
The resulting \QPIP\ protocol also
requires many rounds of quantum communication.

\vb\textbf{Polynomial codes based \QAS\ and its \QPIP\ }
Our second type of \QPIP\ uses a
\QAS\ due to Ben-Or, Cr\'epeau, Gottesman, Hassidim and Smith
\cite{benor2006smq}. This \QAS\ is based on signed quantum polynomial codes,
which are quantum polynomial codes \cite{aharonov1997ftq}
of degree at most $d$ multiplied by
some random sign ($1$ or $-1$) at every
coordinate (this is called the sign key)
and a random Pauli at every coordinate (the pauli key).

We present here a security proof which was missing from the
original paper \cite{benor2006smq}. The proof requires some care,
due to a subtle point, which was not addressed in \cite{benor2006smq}.
We first prove that no Pauli attack can fool more than a small
fraction of the sign keys, and thus the sign key suffices in order to protect
the code from any Pauli attack. Next, we need to show that the scheme
is secure against general attacks. This, surprisingly, does not follow
by linearity from the security against pauli attacks (as is the case
in quantum error correcting codes):
if we omit the Pauli key we get an authentication
scheme which is secure against Pauli attacks but not against general
attacks. We proceed by showing, (with some similarity to
the Clifford based \QAS), that the random Pauli key
effectively translates Eve's attack to a mixture
(not necessarily uniform like in the Clifford
case) of Pauli operators acting on a state encoded by a random signed
polynomial code.

\ignore{
Intuitively, an authentication scheme is
required to be able to
detect \emph{any} intervention with high probability. Let us
compare it to an error detection code. An error detection code,
detects a \emph{certain set} of errors \emph{with certainty},
while other errors always escape detection. The main idea behind
the authentication scheme of \cite{benor2006smq}
is to use a random error detection
code such that \emph{any error} is detected by all but a small
fraction of codes in the set.}

\ignore{
\begin{thm}\label{thm:securityof polynomial}
The Polynomial authentication scheme is secure against
general attacks with security $2^{-d}$\label{thm:PolynomialAuth}, where $d$ is the degree of the polynomial error correction code used in the scheme.
\end{thm}

We note that the security parameter $d$ is different than the $d$ parameter in
the Clifford based \QAS.  The security of the signed polynomial codes \QAS\ is
in fact slightly worse than that of the Clifford based \QAS\ (see
\Sec{sec:authentication} for more precise statements).}

\ignore{
There is a similarity between the Clifford and polynomial \QAS s, which makes
the ideas underlying the proofs quite reminiscent of each other. The similarity
lies in the role of the random Clifford and random Pauli operations.
Intuitively, the random operations' role is to protect some underling structure
which holds the authenticated information. Both type of operations (Clifford,
Pauli) serve to randomize the attack of Eve on the underlying structure. The
random Clifford operation disentangles and symmetrizes any attack, such that
Eve's attack has the effect on the protected structure as a randomly chosen
Pauli. On the other hand, the random Pauli does not symmetrize Eve's attack but
rather disentangles it, such that Eve's attack is reduced to a mixture (not
necessarily uniform) of Pauli operators on the protected structure.  The
underling structure of both \QAS\ makes it possible to detect such
interventions
with high probability.
}

\ignore{
\subsection{The \QPIP\ Protocols}
Given those \QAS s, we can prove \Th{thm:qcircuit} in two ways,
as follows. We start with the proof using the Clifford authentication scheme,
in which case the construction is extremely simple.
In this case, the prover in the \QPIP\ serves merely as
 an untrusted storage device.  He keeps the qubits in the correct
 superposition
  of a state encoded according to the \QAS. To apply the
operations of the
  quantum circuit, $\mcP$ sends the appropriate encoded
 qubits back to $\mcV$
  who, knowing the encoding, can decode the qubits, apply the operations,
 encode them again, and send them back. This gives a scheme which
requires the register to hold $3$ qubits. }

\ignore{
The size of the quantum register which the verifier needs to
possess in this protocol is that of two registers,
which in our case can be of two qubits each, and so we get a register of
$4$ qubits. In fact,
$3$ qubits suffice for the verifier's register,
since the length of the encoding of one qubit is $2$,
and it suffices to send one register at a time and
wait until the verifier decodes it before sending the
second register.
}

\ignore{
The second proof of \Th{thm:main},
based on the polynomial codes
 \QAS, also provides fault tolerance.
The \QPIP\ protocol we use
is very different than that used with the Clifford group \QAS.
}

Due to its algebraic structure, the signed polynomial code \QAS\
allows applying gates without knowing the key.
This was used in \cite{benor2006smq} for
secure multiparty quantum computation; here we use it
to allow the prover to perform gates without
knowing the authentication key.

The \QPIP\ protocol goes as follows. The prover receives
all authenticated qubits in the beginning. Those
include the inputs to the circuit, as well as
authenticated magic states required to perform
Toffoli gates, as described in \cite{benor2006smq,nielsen2000qcq}.
With those at hand,
the prover can perform universal computation using
only Clifford group operations and measurements
(universality was proved for qubits in \cite{bravyi2005uqc},
and the extension to higher dimensions was used in \cite{benor2006smq}).
The prover sends the verifier results of measurements and the verifier
sends information given those results, which enables the
prover to continue the computation. The communication is thus classical
except for the first round.

\ignore{
We mention that the resulting \QPIP\ which uses
the polynomial code \QAS\ has worse parameters
than the Clifford based one:
The size of the register is $9$ qutrits (three dimensional systems)
rather than $3$ qubits.
The reason is that we need three registers of
three qutrits each, to apply the Toffoli gate. In fact, we suspect that
it suffices that the verifier's register contains only
$3$ qutrits, based on ideas generalizing the
result of Bravyi and Kitaev \cite{bravyi2005uqc}
to qudits, namely, to higher dimensions.
This will be discussed in the full version of this paper.
}

\ignore{
This might be a step towards
a scheme involving only classical communication, which is
a major open problem remaining in this paper.
}

\ignore{
\subsection{Fault Tolerant \QPIP}
As discussed before, for the \QPIP\ to be applicable in physically realistic
settings, it must work also in the presence of
noise which affects the communication channels
and the computations performed by the prover and the verifier.
We manage to prove that this indeed holds for the scheme
based on polynomial codes:

\begin{thm}\label{thm:ft}
\Th{thm:qcircuit} holds also when the quantum communication
and computation devices are subjected to the usual local noise model
assumed in quantum fault tolerance settings.
\end{thm}
}

\vb\textbf{Fault Tolerance}
Using the polynomial codes \QAS\ enables applying
known fault tolerance techniques based on polynomial quantum
codes \cite{aharonov1997ftq, benor2006smq} to
achieve robustness to local noise. However, a problem arises when
attempting to apply those directly: in such a scheme,
the verifier needs to send the prover polynomially
many authenticated qubits every time step, so that
the prover can perform error corrections on all qubits
simultaneously. However, the verifier's quantum register contains
only a constant number of qubits, and
so the rate at which he can send authenticated qubits
(a constant number at every time step)
seems to cause a bottleneck in this approach.

We bypass this problem is as follows.
At the first stage of the protocol,
authenticated qubits are sent from the verifier to
the prover, one by one. As soon as the prover receives an
authenticated qubit, he protects his qubits using his own
concatenated error correcting codes so that the effective error
in the encoded authenticated qubit
is constant.
This constant accuracy can be maintained for a long time by
the prover, by performing error correction with respect to
{\it his} error correcting code. Thus, polynomially many
such authenticated states can be passed to the prover
in sequence. A constant effective error is not good enough,
but can be amplified to an arbitrary inverse polynomial by
purification. Indeed, the prover cannot perform
purification on his own since the purification
compares authenticated qubits and the prover does not know the
authentication code; However, the verifier can help in the prover's
using classical communication. This way the prover can reduce the effective
error on his encoded authenticated qubits to inverse polynomial,
and perform the usual fault tolerant construction
of the given circuit, with the help of the prover in performing
the gates.

\vb\textbf{Blind Quantum Computation}
To achieve \Th{thm:blind}, we modify our construction so that the circuit that the prover performs is a {\it universal quantum circuit}, i.e.,
a fixed sequence of gates which
gets as input a description of a quantum circuit, plus an input string to that
circuit, and  applies the input quantum circuit
to the input string.  Since the universal quantum circuit is fixed, it reveals
nothing about the input quantum circuit or the input string to it.

\subsection{Interpretations of the Results}
The corollaries below clarify the connection between the
results and the motivating questions, and show that
one can use the \QPIP\ protocols designed here, to address the
various issues raised in \Sec{sec:funda}.

We start with some basic question. Conditioned that the verifier
does not abort, does he know that
the final state of the machine is very close to
the correct state that was supposed to be
the outcome of the computation?
This unfortunately is not the case. It may be that
the prover can make sure that the verifier
abort with very high probability, but when he does not abort,
the computation is wrong. However a weaker form of the above
result holds: if we know that the probability of not to abort is high,
then we can deduce something about correctness.

\begin{corol} \label{corol:confidence}
For a \QPIP\ protocol with security parameter $\delta$, if the
verifier does not abort  with probability $\ge \gamma$
then the trace distance between the final density matrix
and that of the correct state
is at most $\frac {2\delta} \gamma$
\end{corol}
The proof is simple and is given in \pen{app:interpretation}.

\bigskip
Further interpreting \Th{thm:main}, we show that under a somewhat
stronger assumption than \BQP\ $\ne$ \BPP, but still a widely
believed assumption, it is possible to lower bound the
 computational power of a successful prover and show that it
 is not within \BPP.
Assuming that there is a language $L \in$ \BQP\
and there is a polynomial time samplable distribution $D$ on which any \BPP\  machine errs with non negligible probability (e.g. the standard cryptographic assumptions about the hardness of Factoring or Discrete Log), we have

\begin{corol} For such a language $L$,
if the verifier interacts with a given prover for the language $L$,
and does not abort with high probability, then the prover's
computational power cannot be simulated by a \BPP\ machine.
\end{corol}

This corollary follows immediately from Corollary \ref{corol:confidence}.

One might wonder whether it is possible to somehow get
convinced not only of the fact that the computation
that was performed by the prover is indeed the desired one,
but also that the prover must have had access to some quantum
computer. We prove:

\begin{corol}\label{corol:hybridProver}
  There exists a language $\mcL \in \BQP$ such that even if the prover in our
  \QPIP\ is replaced by one with unbounded classical computational power, but
  only a constant number of qubits, the prover will not be able to convince the
  verifier to accept: $\mcV$ in this case aborts the computation with high
  probability.
\end{corol}

This means that our protocols suggests yet another example in which quantum
mechanics cannot be simulated by classical systems, regardless of how
computationally powerful they are. This property appears already in
bounded storage models \cite{watrous1999sbq}, and of course (in a different
setting) in the \textit{EPR} experiment.


Finally, we remark that in the study of the classical notion of \IP,
a natural question is to ask how powerful the prover must be, to prove
certain classes of languages.
It is known that a \PSPACE\ prover is capable of proving any language
in  \PSPACE, and similarly, it is known that
 \NP\ or \textsf{\#P}
 restricted provers
 can prove any language which they can compute. This is not known for
\textsf{coNP}, \textsf{SZK} or \textsf{PH} \cite{arora:ccm}.
It is natural to ask what is the power of a \BQP\ prover;
our results imply that such a prover can prove the entire
class of \BQP\ (albeit to a verifier who is not entirely classical).
Thus, we provide a characterization of the power of a \BQP\ prover.


\ignore{
Knill \cite{knill2008rbq} has studied independently a related question of
providing methods to test the fact that \emph{small} quantum systems
(of up to ten qubits, say) indeed
constitute a quantum register as they are supposed to. He suggests tests to do
this, based on random Clifford operations and quantum algorithms.  However, his
methods are of a more physics flavor, and they
do not address the fundamental questions presented here.}

\ignore{
As far as we know, our scheme is the first
fault tolerant quantum interactive proof result, though
related notions were discussed by Terhal etc \dnote{complete and check}.}

\subsection{Related Work and Open Questions}
The two questions regarding the cryptographic angle were asked by Childs in
\cite{childs2001saq}, and by Arrighi and Salvail in \cite{blind}, who proposed
schemes to deal with such scenarios.  However \cite{childs2001saq} do not deal
with a cheating prover, and \cite{blind} deals with a restricted set of
functions that are classically verifiable.

After  deriving  the results  of  this work,  we  have  learned that
Broadbent, Fitzsimons,  and  Kashefi  \cite{broadbent2008ubq}  have proven
related results. Using measurement based quantum computation, they
construct a protocol for universal blind quantum computation.
In their case, it suffices that the verifier's
register consists of a single qubit. Their results imply
similar implications to ours, though these are implicit in
\cite{broadbent2008ubq}.

An important and intriguing open question is whether it is possible to
remove the necessity for even a small quantum register, and
achieve similar results in the more natural \QPIP\ model
in which the verifier is
entirely classical. This would have interesting fundamental implications
regarding the ability of a classical system to learn and test a quantum
system.

Another interesting (perhaps related?) open question is to study the model we
have presented here of \QPIP, with more than one prover. Possibly, multiprover
\QPIP\ might be strong enough even when restricted to classical communication.

This work also raises some questions in the philosophy of science.
In particular, it suggests the possibility of formalizing, based on
computational complexity notions,
the interaction between physicists and Nature; perhaps the evolution of
physical theories..
Following discussions with us at preliminary stages of this work,
Jonathan Yaari is currently studying ``interactive proofs with Nature''
from the philosophy of science aspect \cite{JonathanThesis}.

\vb\textbf{Paper Organization}
We start by some notations and background in \Sec{Back}. In
\Sec{sec:authentication} we present both our \QAS\ and prove their security.
In \Sec{sec:IPQ} we present our \QPIP\ protocols together some aspects of their security proofs. Other proofs are delayed to the appendices due to lack
of space:
Fault tolerance is explained in \pen{app:ft}.
Blind delegated quantum computation is proved in \pen{app:blind}.
The corollaries related to the interpretations of the results
are proven in \pen{app:interpretation}.

\section{Background}\label{Back}
\subsection{Pauli and Clifford Group}\label{sec:PauliNClifford}
Let $\mbP_n$ denote the $n$-qubits Pauli group.
$P=P_1\otimes P_2\odots  P_n$ were $P_i \in \{\mcI,X,Y,Z\}$.

\begin{deff}
Generalized Pauli operator over $F_q$:
$X \ket{a} = \ket{\left(a+1\right)\mod q}, ~
Z \ket{a} = \omega_q^a\ket{a}, ~Y = XZ,$
where $\omega_q = e^{2\pi i /q}$ is the primitive q-root of the unity.
\end{deff}
We note that $ZX=\omega_q XZ$.
We use the same notation, $\mbP_n$, for the standard and generalized
Pauli groups, as it will be clear by context which one is being used.

\begin{deff}
For vectors $x,z$ in $F_q^m$, we denote a $P_{x,z}$ the Pauli
operator $Z^{z_1}X^{x_1}\odots Z^{z_m}X^{x_m}$.
\end{deff}

We denote the set of all unitary matrices over a vector space $A$ as
$ \mathbbm{U}(A)$. The Pauli group $\mcP_n$ is a basis to the
matrices acting on n-qubits. In particular, we can write any matrix
$U \in \mathbbm{U}(A\otimes B)$ for $A$ the space of $n$ qubits,
as $\sum_{P\in \mcP_n}P\otimes U_P$ with $U_P$ some matrix on
$B$.

Let $\mfC_n$ denote the $n$-qubit Clifford group. Recall that it is
a finite subgroup of $\mathbbm{U}(2^n)$ generated by the Hadamard matrix-H, by $K=\left(\begin{array}{ll}
1&0\\0&i\end{array}\right)$, and by controlled-NOT.
The Clifford group is characterized by the property that it maps the
Pauli group $\mbP_n$ to itself, up to a phase $\alpha\in\{\pm 1,\pm i\}$. That is:
$\forall C\in\mfC_n ,  P\in \mbP_n: ~\alpha CPC^\dagger \in \mbP_n$

\begin{fact}\label{fa:randomclifford}
A random element from the Clifford group on $n$ qubits can be
sampled efficiently by choosing a string $k$ of $poly(n)$ length uniformly
at random. The map from $k$ to the group element represented as a product
of Clifford group generators can be computed in classical polynomial
time.
\end{fact}

\subsection{Signed Polynomial Codes}\label{sec:SignedPolynomial}
For background on polynomial quantum codes see \pen{app:poly}.
\begin{deff}\label{def:SignedPolynomial} {\bf (\cite{benor2006smq})}
  The signed polynomial code with respect
 to a string $k\in\{\pm 1 \}^m$ (denoted  $\mcC_k$) is defined by:
\begin{equation}
\ket{S^k_a} \EqDef \frac{1}{\sqrt{q^d}}\sum_{f:def(f)\le d ,
f(0)=a}\ket{k_1\cdot f(\alpha_1){\ldots} k_m\cdot f(\alpha_m)}
\end{equation}
\end{deff}
We use $m=2d+1$. In this case, the code can detect
$d$ errors. Also,  $\mcC_k$
 is self dual \cite{benor2006smq}, namely,
the code subspace is equal to the dual code subspace.


\section{Quantum Authentication}\label{sec:authentication}
\subsection{Definitions}

\begin{deff} \label{def:qas} (adapted from Barnum et. al. \cite{barnum2002aqm}).
  A quantum authentication scheme (\QAS) is a pair of polynomial time quantum
  algorithms $\mcA$ and $\mcB$ together with a set of classical keys $\mcK$ such
  that:
\begin{itemize}
\item $\mcA$ takes as input an m-qubit message system M and a key
$k\in\mcK$ and outputs a transmitted system $T$ of $m+d$ qubits.
\item $\mcB$ takes as input the (possibly altered) transmitted system
$T'$ and a classical key $k \in \mcK$ and outputs two systems: a
$m$-qubit message state $M$, and a single qubit $V$ which
indicate whether the state is considered valid or erroneous. The basis states of
$V$ are called $\ket{VAL},\ket{ABR}$. For a fixed
$k$ we denote the corresponding super-operators by $A_k$ and
$B_k$.
\end{itemize}
\end{deff}

\ignore{
Note that $\mcB$ may well have measured the qubit $V$, but we would rather
keep the quantum description so that we can use density matrices.

There are two conditions which should be met by a quantum
authentication protocol. On the one hand, in the absence of
intervention, the received state should be the same as the initial
state and $\mcB$ should not abort.

On the other hand, we want that when the adversary does intervene,
$B$'s output systems have high fidelity to the statement ``either
$\mcB$ rejects or his received state is the same as that sent by
$\mcA$''. This is formalized below for pure states; one can deduce
the appropriate statement about fidelity of mixed states, or for
states that are entangled to the rest of the world (see
\cite{barnum2002aqm} Appendix B).
}

Given a pure state $\ket{\psi}$, consider the following test on
the joint system $M, V$ : output a $1$ if the first $m$ qubits are in
state $\ket{\psi}$ or if the last qubit is in state $\ket{ABR}$,
otherwise, output  $0$. The corresponding projections are:
\begin{eqnarray}
P_1^{\ket{\psi}} & = &\ket{\psi}\bra{\psi} \otimes I_V+ \left(I_M- \ket{\psi}\bra{\psi}\right) \otimes
\ket{ABR}\bra{ABR} \\
P_0^{\ket{\psi}} & = &(I_M -
\ket{\psi}\bra{\psi}) \otimes \ket{VAL}\bra{VAL}
\end{eqnarray}
The scheme is secure if for all possible input states $\ket{\psi}$ and for all
possible interventions by the adversary, the expected fidelity of
$\mcB$'s output to the space defined by $P^{\ket{\psi}}_1$ is high:

\begin{deff}
A \QAS\ is secure with error $\epsilon$ if for every state $\ket\psi$
it holds:
\begin{itemize}
\item
Completeness: For all keys $k\in \mcK$ :
$B_k(A_k(\ket{\psi}\bra{\psi})) = \ket{\psi}\bra{\psi}\otimes
\ket{VAL}\bra{VAL}$
\item
Soundness: For any super-operator $\mcO$ (representing a possible intervention by the adversary), if $\rho_B$ is defined by
defined by $\rho_B = \frac{1}{|\mcK|}\sum_k B_k\big(\mcO(A_k(\ket{\psi}\bra{\psi}))\big)$, then:\; $\tr{(P_1^{\ket{\psi}}\rho_B)} \ge 1- \epsilon$.
\end{itemize}
\end{deff}

\subsection{Clifford Authentication Scheme}\label{sec:CliffordAuth}

\begin{protocol}\textbf{Clifford based \QAS\ }:
  Given is a state $\ket\psi$ on $m$ qubits and $d\in\mN$ a security
  parameter. We denote $n=m+d$. The set of keys $\mcK$ consists of succinct
  descriptions of Clifford operations on $n$ qubits
(following Fact \ref{fa:randomclifford}).
We denote by $C=C_k$ the operator specified by a
  key $k\in \mcK$.
\begin{itemize}
\item \textbf{Encoding - $A_k$}: Alice applies $C_k$ on the state $\ket\psi \otimes
\ket{0}^{\otimes d}$.
\item \textbf{Decoding - $B_k$}: Bob applies $C_k^\dagger$ to the received
      state.  Bob measures the auxiliary registers and declares the state valid
      if they are all $0$, otherwise Bob aborts.
\end{itemize}
\end{protocol}

\begin{thm}{\label{thm:CliffordAuth}}
  The Clifford scheme applied to $n=m+d$ qubits is a \QAS\ with security
  $2^{-d}$. Where $d$ is the number of qubits added to a message on $m$ qubits.
\end{thm}

\begin{proof} \textit{Sketch.(The full proof is given in \pen{app:cliffordsecurity}). }
We show that when Eve applies a non
trivial Pauli operator, then averaging over the random Clifford operators, the
effective transformation on the original state is as an application of a random
Pauli. Hence, any Pauli attack is detected with high probability. We then
show that \emph{any} attack of Eve is reduced to a very specific form:
\begin{equation}\label{eq:Mdef}\mcM_s: \rho \rightarrow s\rho+ (1-s)\frac 1
{4^n-1}\sum_{P\ne \mcI} P\rho P^\dagger
\end{equation}
(for some $0\le s\le 1$). It is not
hard to see, using linearity, that this type of attack is detected
with high probability.
\end{proof}

\ignore{
Essentially, we show that the average of a conjugation of any operator by
Clifford operators is equal to the mixing operator $\mcM_s$ for some $s$;
furthermore, we show that these type of attacks are detected with high
probability.}

Given $r$ blocks of $m$ qubits each, we can apply the
\QAS\ separately on each one of the $r$ blocks.
$\mcB$ declares the state valid  if all
of the $r$ registers are valid according to the original Clifford \QAS\ .
We call this the {\it concatenated Clifford protocol}.
The completeness of the concatenated protocol is trivial,
reasoning as in the original \QAS. For soundness we have the following
theorem, whose proof is given in \pen{app:concatclifford}.

\begin{thm} \label{thm:CliffordConcat} The concatenated Clifford protocol has
  the security of the individual Clifford with security parameter $d$, \QAS,
  that is $2^{-d}$. This holds regardless of the number of blocks ($r$) that are
  authenticated.
\end{thm}

\subsection{Polynomial Authentication Scheme}
\label{sec:PolynomialAuth}

\begin{protocol} \textbf{Polynomial Authentication protocol }:
Alice wishes to send the state $\ket\psi$ of dimension $q$. She
chooses a security parameter $d$, and a code length $m=2d+1$.
\begin{itemize}
\item \textbf{Encoding:}
 Alice randomly selects a pair of keys: a sign key $k\in\{\pm
1\}^m$ and a Pauli key $(x,z)$ with $x,z \in {F_q}^m$. She
 encodes $\ket\psi$ using the signed quantum polynomial code
$\mcC_k$ of polynomial degree $d$ (see \Def{def:SignedPolynomial}).
She then applies the Pauli $P_{(x,z)}$ (i.e., for $j \in \{1,..,m\}$ she applies  $Z^{z_j}X^{x_j}$ on the $j$'th qubit).

\item \textbf{Decoding}
 Bob applies the inverse of
$P_{(x,z)}$, and performs the error detection procedure of the code $\mcC_k$.
He aborts if any error is found and declares the message valid  otherwise.
\end{itemize}
\end{protocol}

The completeness of this protocol is trivial. We proceed to prove the security of the protocol.

\ignore{if Eve does not interfere with the state then Bob recovers
exactly the original
state and never aborts.}

\begin{thm}{\label{thm:PolynomialAuth}} The polynomial authentication scheme
  is secure against general attacks with security~$2^{-d}$
\end{thm}

\begin{proof} A sketch was given in the introduction;
the full proof is given in
\pen{app:securitypolynomial}.
\end{proof}

We notice that in this scheme a
$q$-dimensional system is encoded into a system of
dimension $q^m=q^{2d+1}$. The same security is achieved in the Clifford
 \QAS\ by encoding $q$ into $q\cdot 2^d$ dimensions.
 The polynomial scheme is somewhat worse
in parameters, but still with an exponentially good security.

To encode several registers, one can independently authenticate
each register as in the Clifford case,
(\Th{thm:CliffordConcat}) but in fact we can use the same sign key
$k$ for all registers, while still maintaining security.
This fact will be extremely useful in \Sec{sec:IPQ}.
The following theorem is proved in \pen{app:concatpolynomial}.

\begin{thm} \label{thm:concatpolynomial} The concatenated polynomial based \QAS\ (with the same sign key for all registers), and with degree $d$ polynomial,  has
the same security as the individual \QAS, that is: $2^{-d}$.
\end{thm}


\section{Interactive Proof For Quantumness}\label{sec:IPQ}

\subsection{Clifford Authentication Based Protocol}\label{sec:cliffordIP}

\begin{protocol} \textbf{Clifford based Interactive Proof for \qc}:
  \label{prot:CliffordIP}Fix a security parameter $\epsilon$.  Given is a
  quantum circuit consisting of two-qubit gates, $U=U_T{\ldots}U_1$, with error
  probability reduced to $\le \delta$.  The verifier authenticates the input
  qubits of the circuit one by one using the (concatenated) Clifford \QAS\ with
  security parameter $d = \lceil{\log{\frac 1 \epsilon}}\rceil$, that is every
  qubit is authenticated by $d+1$ qubits, and sends them to $\mcP$.  For each
  $i=1$ to $m$, the verifier asks the prover for the authenticated qubits on
  which he would like to apply the gate $U_i$, decodes them, aborts if any error
  is found, applies the gate, authenticates the resulting qubits using a new
  pair of authentication keys, and sends the encoded qubits back to
  $\mcP$. Finally, the verifier asks $\mcP$ to send the output authenticated
  qubit, decodes and aborts if any error is found; otherwise, measures the decoded
  qubit and accepts or rejects accordingly. \label{prot:finalMeasurement} In any
  case that $\mcV$ does not get the correct number of qubits he aborts.
\end{protocol}

\begin{statement}{\Th{thm:qcircuit}} For any $\epsilon,\delta >0$
\Prot{prot:CliffordIP} is a \QPIP\ protocol
  with completeness $1-\delta$ and soundness $\delta+\epsilon$ for
  \qc.
\end{statement}

\begin{proof}
  If the prover is honest, the verifier will declare valid with certainty.
  Since the error in the circuit is $\le \delta$, $(1-\delta)$ completeness
  follows. For soundness, we observe that for the verifier to accept if $x$ is
  not in the language, means that he has not aborted, and also, answers YES. Let
  us denote by $P_{bad}$ the projection on this subspace (\emph{Valid} on the first
  qubit, \emph{Accept} on the second).  To bound the probability of this event, we
  observe that the correct state at any given step is a state which is
  authenticated by the concatenated Clifford \QAS. We can thus use the
  decomposition of Eve's attacks to Paulis, namely \Eq{attackProfile}.
Observing
  that a Pauli attack in our scheme is either declared valid or leads to abort,
  implies that the final density matrix can be written
  as
  \begin{equation}\rho_{final}=(\alpha_0\rho_0
    +\alpha_{c}\rho_c)\otimes\ket{VAL}\bra{VAL}
    +\alpha_{1}\bar\rho_1\otimes\ket{ABR}\bra{ABR},
  \end{equation}
where $\rho_c$ is the correct state.
  To bound $\tr(P_{bad}\rho_{final})$ we observe that the left term in bounded
  by the security parameter of the \QAS, namely $\epsilon$, the second term is
  bounded by the error caused by the quantum circuit, namely $\delta$, and the
  third term vanishes.
\end{proof}

The classical communication is linear in the number of gates.  For $\epsilon
=\frac 1 2$, we get $d=1$, and so the verifier uses a register of $4$ qubits.
In fact $3$ is enough, since each of the authenticated qubits can be decoded (or
encoded and sent) on its own before a new authenticated qubit is handled.


\subsection{Polynomial Authentication Based Protocol}
We start by describing how the prover performs
a set of universal gates on authenticated qubits, using classical
communication with the verifier, and special states called
Toffoli states. This set of operations, namely Clifford group operations
augmented with the Toffoli gate, form a universal set
of gates \cite{benor2006smq}.

\vb\textbf{Application of Quantum Gates} \label{des:secapp}
We denote encoded gates (logical operators) with a tilde.
For the full description of how to apply each of these logical gates
see \pen{app:polynomialgates}.  Briefly,
for Pauli operators, the verifier merely updates his Pauli key.
For the control-SUM, and the Fourier transform, the prover applies the
gates transversally as if the code was the standard polynomial codes,
and the verifier updates his sign and Pauli keys.
For the measurement, the prover measures the register, sends the result to
the verifier, who returns its interpretation which he computes using his
keys. The Toffoli gate is applied using the above, on the relevant
authenticated qubits plus an authenticated Toffoli state \cite{benor2006smq}.

\begin{protocol}\textbf{Polynomial based Interactive Proof for
    \qc} \label{prot:PolynomialIP}
Fix a security parameter $\epsilon$.  Given is a
 quantum circuit on $n$ qubits
consisting gates from the above universal set, $U=U_T{\ldots}U_1$.
We assume the circuit has error
 probability $\le \delta$.
The verifier sets  $d = \lceil{\log{\frac 1 \epsilon}}\rceil$ and
 uses $3$ registers of $m=2d+1$ qudits each, where each qudit is
of dimensionality $q>m$.  The verifier uses concatenated
polynomial \QAS\ with security parameter $d$ to authenticate
$n$ input qudits and the necessary number of Toffoli states.
This is done sequentially using $3m$ qudits at a time.
Then, the prover and verifier perform the gates of the circuit
as described above. Finally,
 if the final measurement does not yield an authenticated answer,
      the verifier\textbf{aborts}, otherwise, he accepts or rejects
according to the measurement outcome.
\end{protocol}

\begin{thm}\Prot{prot:PolynomialIP} is a \QPIP\ protocol with
  completeness $1-\delta$ and soundness $\delta +\epsilon$ for \qc.
\label{thm:PolynomialIP}
\end{thm}

This theorem implies a second proof for \Th{thm:qcircuit}.
The size of the verifier's register is naively $3m$, but using the same idea
as in the Clifford case, $m+2$ suffice. With $\epsilon=1/2$, this gives a
register of $5$ qutrits.

\begin{proof} \textit{(Sketch. The full proof can be found in \pen{apen:polyIPproof})}
 The completeness is trivial, similarly to the Clifford case.
To prove the soundness of the protocol we first prove the following lemma.
\begin{lem}\label{lem:uniformkeys}
At any stage of the protocol the verifier's set of keys,
$k$ and $\{(x,z)_i\}_1^n$ are distributed uniformly and
independently.
\end{lem}
This implies that the correct state in an encoded states according to
the concatenated \QAS. The rest of the argument follows closely that
of the proof of \Th{thm:qcircuit}.
\end{proof}


\section{Acknowledgements}
D.A. thanks Oded Goldreich, Madhu sudan and Guy Rothblum
for exciting and inspiring conversations that eventually led to this work.
E.E. thanks Avinatan Hassidim for
stimulating and refining ideas, particularly about fault tolerance.
We also thank Gil Kalai, David DiVincenzo and Ari Mizel, for stimulating
questions and clarifications, and Daniel Gottesman for many helpful ideas and
remarks, and in particular, for his help in
proving \Th{thm:CliffordAuth}.

\bibliographystyle{alpha}
\bibliography{qpip}
\newpage \appendix

\section{Polynomial Quantum Error Correction Codes}\label{app:poly}\begin{deff} Polynomial error correction code
  \cite{aharonov1997ftq}. Given $m,d,q$ and $\{\alpha_i\}^m$ where $\alpha_i$
  are distinct non zero values from $F_q$, the encoding of $a\in F_q$ is
  $\ket{S_a}$
\begin{equation}
\ket{S_a} \EqDef \frac{1}{\sqrt{q^d}}\sum_{f:def(f)\le d ,
f(0)=a}\ket{ f(\alpha_1),{\ldots} , f(\alpha_m)}
\end{equation}
\end{deff}

We use here $m=2d+1$, in which case the code subspace is its own dual.
It is easy to see that this code can detect up to $d$ errors
\cite{aharonov1997ftq}.
It will be useful to explicitly state the logical
gates of \textit{SUM}, Fourier ($F$) and Pauli operators ($X,Z$). We will
see that it is possible to apply the logical operations of the
Pauli operators or the controlled-sum by a simple transitive operation. We
can easily verify that applying $X^{\otimes m}$ is the logical $\wt X$
operation:
\begin{equation}\begin{split}
\wt{X} \ket{S_a} =& X^{\otimes m}  \frac{1}{\sqrt{q^d}}\sum_{f:def(f)\le d ,
f(0)=a}\ket{ f(\alpha_1),{\ldots} , f(\alpha_m)} \\
=&  \frac{1}{\sqrt{q^d}}\sum_{f:def(f)\le d ,
f(0)=a}\ket{ f(\alpha_1)+1,{\ldots} , f(\alpha_m)+1}
\end{split}
\end{equation}
setting $f'(\alpha) = f(\alpha)+1$
\begin{equation}\begin{split}
{\ldots} =& \frac{1}{\sqrt{q^d}}\sum_{f':deg(f')\le d ,
f'(0)=a+1}\ket{ f'(\alpha_1),{\ldots} , f'(\alpha_m)}\\
=&\ket{S_{(a+1)}}
\end{split}\end{equation}
Similarly for logical \textit{SUM} , we consider the transitive
application of controlled-sum, that is a  \textit{SUM} operations applied
between the $j$'th register of $\ket{S_a}$ and $\ket{S_b}$.
\begin{equation}\begin{split}
\wt{\textit{SUM}} \ket{S_a} \ket{S_b} =\ & (\textit{SUM})^{\otimes
m}\frac{1}{{q^d}}\sum_{ f(0)=a}\ket{ f(\alpha_1),{\ldots} ,
f(\alpha_m)} \sum_{ h(0)=b}\ket{ h(\alpha_1),{\ldots} ,h(\alpha_m)}\\
=& \frac{1}{{q^d}}\sum_{f(0)=a,h(0)=b}\ket{ f(\alpha_1),{\ldots} ,
f(\alpha_m)}\ket{ h(\alpha_1)+f(\alpha_1),{\ldots} ,
h(\alpha_m)+f(\alpha_m)}
\end{split}\end{equation}
We set $g(\alpha)= f(\alpha)+h(\alpha)$
\begin{equation}\begin{split}
{\ldots} =& \frac{1}{{q^d}}\sum_{f(0)=a,g(0)=a+b}\ket{ f(\alpha_1),{\ldots} ,
f(\alpha_m)} \ket{ g(\alpha_1),{\ldots} ,
g(\alpha_m)} \\
=& \ket{S_a}\ket{S_{a+b}}
\end{split}\end{equation}

Showing what is the logical Fourier transform on the polynomial code requires more work. We first recall the
definition of the Fourier transform in $F_q$:
\begin{eqnarray}
F \ket{ a} &\EqDef& \frac 1 {\sqrt q} \sum_b \omega_q^{ab}\ket b
\end{eqnarray}
We consider an $r$-variant of the Fourier transform which we denote
$F_r$
\begin{eqnarray}
F_r \ket{a} &\EqDef& \frac 1 {\sqrt q} \sum_b \omega_q^{rab}\ket b
\end{eqnarray}
In addition we need the following claim:
\begin{lem}\label{inter} For any $m$ distinct numbers $\{\alpha_i\}_1^m$ there
exists $\{c_i\}_1^m$ such that \begin{eqnarray}
\sum_{i=1}^m c_if(\alpha_i) = f(0)
\end{eqnarray}
For any polynomial of degree $\le m-1$.
\end{lem}

\begin{proof}
A polynomial $p$ of degree $\le m-1$ is completely determined by it's
values in the points $\alpha_i$. We write $p$ as in the form of the
Lagrange interpolation polynomial:
$f(x) = \sum_i \prod_{j\ne i} \frac{x-\alpha_j}{\alpha_i
-\alpha_j}f(\alpha_j)
$.
Therefore, we set $c_i = \prod_{j\ne i}\frac
{-\alpha_j}{\alpha_i-\alpha_j}$ and notice that it is independent of $p$,
and the claim follows.
\end{proof}

We are now ready to define the logical Fourier transform.
\begin{claim}\label{claim:fourier} The logical Fourier operator $\wt F$ obeys
  the following identity:
\begin{eqnarray}
\wt{F} \ket{S_a} \EqDef F_{c_1}\otimes F_{c_2}\odots F_{c_m}\ket{S_a} =& q^{-m/2}\sum_b
\omega_q^{ab} \ket{\wt{S_b}}
\end{eqnarray}
Where $\wt{S_b}$ is the encoding of $b$ in a polynomial code of
degree $m-d$ on $m$ registers.
\end{claim}
\begin{proofof}{\Cl{claim:fourier}}
We denote $\ket{\bar{f}} = \ket{f(\alpha_1),{\ldots} ,f(\alpha_m)}$
\begin{eqnarray}
F_{c_1}\otimes F_{c_2}{\ldots} \otimes F_{c_m}\ket{S_a} &= &
q^{-d/2} F_{c_1}\otimes F_{c_2}{\odots} F_{c_m}\sum_{f:def(f)\le
d,f(0)=a}\ket{\bar p} \\
& = & q^{-d/2} q^{-m/2}\sum_{f:def(f)\le d,f(0)=a} \sum_{b_1,{\ldots} ,b_m}
\omega_q^{\sum_i c_if(\alpha_i)b_i}\ket{b_1,{\ldots} ,b_m}
\end{eqnarray}
We think of the $b_i$'s as defining a polynomial $g$ of degree $\le
m-1$ that is $g(\alpha_i)=b_i$ and split the sum according to $g(0)$:
\begin{eqnarray}\label{ignore}
{\ldots} & = & q^{-(m+d)/2} \sum_{f:def(f)\le d,f(0)=a} \sum_b
\sum_{g:deg(g)\le m-1,g(0)=b }
\omega_q^{\sum_i c_if(\alpha_i)g(\alpha_i)}\ket{\bar g}\label{pre}
\end{eqnarray}

We temporally restrict our view to polynomials $g$ with degree at most
$m-d-1$ and therefore  the polynomial $fg$ has degree at most $m-1$.
We use \Le{inter} on $fg$:
\begin{eqnarray}
\sum_{i=1}^m c_i (fg)(\alpha_i) &= fg(0) &=ab
\end{eqnarray}
Going back to \Eq{pre}:
\begin{eqnarray}
q^{-(m+d)/2} \sum_{p,g} \sum_{b\in F_q} \omega_q^{\sum_i c_i
(fg)(\alpha_i)}\ket{\bar g}  &=& q^{-(m+d)/2}
\sum_{b\in F_q}\sum_{f,g} \omega_q^{ab}\ket{\bar g}
\end{eqnarray}
Where the summation is over all $f,g$ such that $f(0)=a$ and
$g(0)=b$ while the degrees of $f$ and $g$ are at most $d$ and
$m-d-1$ respectively.

The sum does not depend on $f$ and there are exactly $q^d$
polynomials $f$ in the sum, therefore, we can write the expression as :
\begin{equation}\begin{split}
{\ldots} =&\quad  q^{-(m+d)/2} \sum_{b\in F_q}q^d\sum_{g}
\omega_q^{ab}\ket{\bar g} \\
\ =&\quad  \frac{ 1}{ \sqrt{q}} \sum_{b\in F_q} \omega_q^{ab}\frac {1}
{ \sqrt{q^{m-d-1}}} \sum_{g:deg(g)\le
m-d-1, g(0)=b}\ket{\bar g} \\
\ =&\quad  \frac 1 {\sqrt{q}} \sum_{b\in F_q} \omega_q^{ab}\ket{\wt{S_b}}
\end{split}\end{equation}
Since the above expression has norm 1, if follows that the
coefficients that we temporally ignored at \Eq{ignore} all vanish.
\end{proofof}
\begin{corol}
If $m=2d+1$ then it follows from \Cl{claim:fourier} that the code is self dual.
\end{corol}
\begin{claim}
The logical Pauli $Z$ operator $\wt{Z}$ is $Z^{c_1} \odots
Z^{c_m}$.
\end{claim}
The proof of this claim is omitted since it is extremely
similar to the proof of \Cl{claim:fourier}.

\section{Clifford Authentication Scheme}
\subsection{Security Proof of Clifford \QAS}\label{app:cliffordsecurity}

\begin{proofof}{\Th{thm:CliffordAuth}}
  We denote the space of the message sent from Alice to Bob as $M$.  Without
  loss of generality, we can assume that Eve adds to the message a system $E$
  (of arbitrary dimension) and performs a unitary transformation $U\in
  \mathbbm{U}(M\otimes E)$ on the joint system.  We note that there is a unique
  representation of $U=\sum_{P\in \mbP_n}P\otimes U_P$ since the Pauli matrices
  form a basis for the $2^n\times 2^n$ matrix vector space. We first
  characterize the effect that Eve's attack has on the unencoded message:
  $\ket\psi\otimes{\ket 0}^{\otimes d}$.

\begin{claim}\label{claim:EveAttack}
  Let $\rho=\ket\psi\otimes{\ket 0}^{\otimes d}$ be the state of Alice before
  the application of the Clifford operator. For any attack
  \mbox{$U=\sum_PP\otimes U_P$} by Eve, Bob's state after decoding is
  $\mcM_s(\rho)$, where $s=\tr{(U_\mcI^{\dagger}U_\mcI)}$.
\end{claim}

We proceed with the proof of the theorem. From the above claim we know what
Bob's state after Eve's intervention is and we would like to bound its projection on $P_1^{\ket\psi}$:
\begin{eqnarray}\label{tr}
\tr{\Big(P_1^{\ket\psi} \big(s\rho + \frac{1-s}
{4^n -1}\sum_{Q\in\mbP_n\setminus\{\mcI\}}Q\rho
Q^\dagger\big)\Big)}
&= &s\tr(P_1^{\ket\psi}\rho) + \frac{1-s}{4^n-1}
\sum_{Q\in\mbP_n\setminus\{\mcI\}}\tr{(P_1^{\ket\psi}
Q\rho Q^\dagger)}\label{eq:ClifSec}
\end{eqnarray}
By definition of $P_1^{\ket\psi}$ we see that $\tr(P_1^{\ket\psi}\rho)=1$. On
the other hand: $\tr{(P_1^{\ket\psi}Q\rho Q^\dagger)}=1$ when $Q$ does not flips
any auxiliary qubit, and vanishes otherwise.  The Pauli operators that do not
flip auxiliary qubits can be written as $Q'\otimes Q''$ where \mbox{$Q'\in
  \mbP_m$} and \mbox{$Q'' \in \{\mcI,Z\}^{\otimes d}$}. It follows that the
number of such operator is exactly $4^m2^d$. Omitting the identity $\mcI_n$ we
are left with $4^m2^d-1$ operators which are undetected by our scheme. We return
to \Eq{eq:ClifSec}:
\begin{eqnarray}
{\ldots}& \ge &s +\left(1-s\right)(1-\frac{4^m2^d-1}{4^n-1})\\
& \ge &s +\left(1-s\right)(1-\frac{4^m2^d}{4^{m+d}})\\
& = &1 -\frac{1-s} {2^d} \label{goodProjection}
\end{eqnarray}
The security follows from the fact that $s\ge 0$, and hence the
projection is bounded by $1 -\frac{1} {2^d}$.

\end{proofof}

We remark that the above proof in fact implies a stronger theorem:
interventions that are very close to $\mcI$ are even more
likely to keep the state in the space defined by $P_1^{\ket\psi}$.

What remains to prove is \Cl{claim:EveAttack} which is stated above. To this end
we need three simple lemmata:

\begin{lem} \label{mix} Fix a non-identity Pauli operator.  Applying a random
  Clifford operator (by conjugation) maps it to a Pauli operator chosen
  uniformly over all non-identity Pauli operators. More formally, for every
  $P,Q\in\mbP_n\setminus\{\mcI\}$ it holds that : $\left|\left\{C\in\mfC_n |
      C^\dagger PC =Q\right\}\right| = \frac {\left|\mfC_n\right|} {\left|\mbP_n
    \right| -1}= \frac {\left|\mfC_n\right|} {4^n -1}$.
\end{lem}

\begin{lem} \label{clifTw} Let $P\ne P'$ be Pauli operators. For
  any $\rho$ it holds that: \mbox{$\sum_{C\in\mfC_n}C^\dagger P C\rho C^\dagger
    P'C=0$}.
\end{lem}

\begin{lem} \label{Decompose} Let $U=\sum_{P\in\mbP_n} P\otimes U_P$ be a
  unitary operator. For any  density matrix $\rho$:
\begin{equation}\sum_{P\in\mbP_n}
\tr(U_P\rho U_P^{\dagger}) =1
\end{equation}
\end{lem}
Assuming these lemmata we are ready to prove the claim:

\begin{proofof}{\Cl{claim:EveAttack}}
Let $ U = \sum_{P\in \mbP_n} P\otimes U_P$ be the operator applied
by Eve. We denote $\rho=\ket\psi \bra\psi\otimes\ket {\bar
0}\bra{\bar 0}$ the state of Alice prior to encoding.  Let us now
write the state of Bob's system after decoding and before measuring
the $d$ auxiliary qubits.  For clarity of reading we omit the
normalization factor $|\mfC_n|$ and denote the Clifford operation
applied by Alice (Bob) $C$ ($C^\dagger$):
\begin{eqnarray}
\rho_{Bob}&=& \tr_E{\Big(\sum_{C\in\mfC_n}(C\otimes \mcI_E)^\dagger
U\left( (C\otimes \mcI_E)\rho (C\otimes \mcI_E)^\dagger \otimes
\rho_E\right)U^\dagger(C\otimes \mcI_E)\Big)} \\
&=&
\tr_E{\Big(\sum_{P,P'\in\mbP_n}\sum_{C\in\mfC_n }(C\otimes \mcI_E)^\dagger
P\otimes U_P\left( (C\otimes \mcI_E)\rho (C\otimes \mcI_E)^\dagger \otimes
\rho_{E}\right) P'\otimes U_{P'}^\dagger(C\otimes \mcI_E)\Big)}
\end{eqnarray}
Regrouping elements operating on $M$ and on $E$ we have:
\begin{eqnarray}\begin{aligned}
{\ldots} =&\
\tr_E{\Big(\sum_{P,P'\in\mbP_n}\sum_{C\in\mfC_n}\left(C^\dagger
P C\rho C^\dagger
P'C\right)\otimes  U_P\rho_{E} U_{P'}^\dagger \Big)}\\
=&
\sum_{P,P'\in\mbP_n}
\sum_{C\in\mfC_n}\left(C^\dagger P C\rho C^\dagger  P'C\right)
\cdot \tr{\big( U_P\rho_{E} U_{P'}^\dagger \big)}
\end{aligned}\end{eqnarray}
We use \Le{clifTw} and are left only with $P=P'$
\begin{eqnarray}\label{prev}
{\ldots}  =
\sum_{P\in\mbP_n}\sum_{C\in\mfC_n}\left(C^\dagger
P C\rho C^\dagger
PC\right)\cdot \tr{\big(  U_P\rho_{E} U_{P'}^\dagger \big)}
\end{eqnarray}
We first consider the case were $P=\mcI$, then:
\begin{equation}
\sum_{C\in\mfC_n}C^\dagger
P C\rho C^\dagger
PC =|\mfC_n|\rho
\end{equation}
On the other hand when, $P\ne\mcI$  by \Le{mix}:
\begin{equation}
\sum_{C\in\mfC_n} C^\dagger
P C\rho  C^\dagger
PC =\sum_{Q\in\mbP\setminus\{\mcI\}}Q\rho Q^\dagger \frac {|\mfC_n|}{|\mbP_n|-1}
\end{equation}
Plugging the above two equations in \Eq{prev}:
\begin{eqnarray}\begin{aligned}
{\ldots} =&\; |\mfC_n|\rho \tr{\big(
U_\mcI\rho_{E}U_{\mcI}^\dagger\big)}&+&
\sum_{P\in\mbP_n\setminus\{\mcI\}}
\sum_{Q\in\mbP_n\setminus\{\mcI\}}\big(Q\rho Q^\dagger \big)\frac
{|\mfC_n|}{|\mbP_n|-1}
\tr{\big(   U_P\rho_{E}  U_{P'}^\dagger\big)}\\
=&\; |\mfC_n|\rho \tr{\big( U_\mcI\rho_{E}U_{\mcI}^\dagger \big)}&+&
\frac {|\mfC_n|\sum_{P\ne\mcI} \tr{\big(U_P\rho_{E}U_{P}^\dagger
\big)}}{|\mbP_n|-1} \sum_{Q\in\mbP_n\setminus\{\mcI\}
}\big(Q\rho Q^\dagger \big)
\end{aligned}\end{eqnarray}
We use \Le{Decompose} and so Bob's state after renormalization can be written as:
\begin{equation}\label{attackProfile}
s\rho + \frac{(1-s)}{4^n-1}\sum_{Q\in\mbP_n\setminus\{\mcI\}}\left(Q\rho Q^\dagger
\right)
\end{equation}
For $s=\tr({ U_{\mcI}\rho U_{\mcI}^\dagger})$, which concludes the proof.
\end{proofof}


Finally, we prove the lemmata stated above:

\begin{proofof}{\Le{mix}}
  We first claim that or every
  $Q,P\in\mbP_n\setminus \mcI$ there exists $D\in\mfC_n$ such that $D^\dagger P
  D=Q$. We will prove this claim by induction. Specifically, we show that
  starting form any non identity Pauli operator one can, using conjunction by
  Clifford group operator reach the Pauli operator $X\otimes \mcI^{\otimes
    n-1}$.

We first notice that the swap operation is in $\mfC_2$ since it holds that:
\begin{eqnarray}
    SWAP_{k,k+1} &= &CNOT_{k\rightarrow (k+1)}CNOT_{(k+1)\rightarrow k}CNOT_{k\rightarrow(k+1)}
\end{eqnarray}
Furthermore, we recall that $K^\dagger(XZ)K\propto X$ and $H^\dagger
ZH=X$. Therefore, any Pauli $P=P_1\odots P_n$ can be transformed using $SWAP,H$
and $K$ to the form: $X^{\otimes k}\otimes \mcI^{\otimes n-k}$ (up to a
phase). To conclude we use:
\begin{eqnarray}
CNOT_{1\rightarrow2}^\dagger (X_1\otimes X_2)CNOT_{1\rightarrow2}& =& X\otimes \mcI
\end{eqnarray}
which reduces the number of $X$ operations at hand. Applying this sufficiently many times results in reaching the desired form. Since this holds for any non-identity Pauli operators: $P,Q$ we know there are $C,D\in \mfC_n$ such that:
\begin{eqnarray}
X\otimes\mcI^{\otimes n-1}&=&C^\dagger PC=D^\dagger QD \\
&\Rightarrow& DC^\dagger PCD^\dagger= Q
\end{eqnarray}
therefore $CD^\dagger$ is the operator we looked for.  We return to the proof of
the Lemma, let us first fix some $Q\ne\mcI$, it will suffice to prove that for
any $P,P'$ the set $A_{P,P'}\EqDef\left\{C\in\mfC_n| C^\dagger PC=P'\right\}$ is
of a fixed size. We set $D\in\mfC_n$ such that $D^\dagger PD=Q$ then it holds
that: $CD\in A_{Q,P'} \iff C\in A_{P,P'}$ therefore $|A_{P,P'}| = |A_{Q,P'}|$,
and $|A_{Q',P'}| = |A_{Q,P}|$ follows trivially.

We use the fact that the sets $\{A_{P,Q}\,:\forall P \}$ is a partition of $\mfC_n$, and that all $A_{P,Q}$ have the same size:
\begin{eqnarray}
\left|\mfC_n\right| =\sum_{P'\in\mbP_n\setminus\mcI}\left|A_{P',Q}\right|
=& (4^n -1)\left|A_{P,Q}\right|
\end{eqnarray}
Which concludes the proof.
\end{proofof}

\begin{proofof}{\Le{clifTw}}
Since $P\ne P'$ we know there exists an index $i$ such that $P_i\ne
P_i'$ that is:
\begin{eqnarray}
P_i = X^aZ^b &\;&
P_i' = X^{a'}Z^{b'}
\end{eqnarray}
where $(a,b)\ne (a',b')$. let us define
$Q_i=X_i^{1-b-b'}Z_i^{1-a-a'}\otimes \mcI$.  We notice
that $(Q_i\otimes\mcI) C\in \mfC_n$ and furthermore any operator in
$\mfC_n$ can be written in this form. We write $Q_iC$ instead of
$(Q_i\otimes\mcI) C$ for simplicity.
\begin{eqnarray}
\sum_{C\in\mfC_n} C^\dagger P C\rho C^\dagger  P'C
& = &
\sum_{Q_i C\in\mfC_n}
\left(Q_i C\right) ^\dagger P \left(Q_i C\right)
\rho \left(Q_i C\right) ^\dagger  P'\left(Q_i C\right)\\
& = &
\sum_{Q_i C\in\mfC_n}
C^\dagger Q_i^\dagger  P Q_i C
\rho C^\dagger Q_i^\dagger  P'Q_i C
\end{eqnarray}
It is easy to check that $Q$ commutes with either $P'$ or $P$ and
anti-commutes with the other. Therefore:
\begin{eqnarray}
{\ldots} & = & (-1) \sum_{Q_i C\in\mfC_n} C^\dagger Q_i^\dagger Q_i
PC \rho C^\dagger Q_i^\dagger Q_i P'C \\
& = & (-1) \sum_{Q_i C\in\mfC_n} C^\dagger
PC \rho C^\dagger P'C\\
& = & (-1) \sum_{ C\in\mfC_n} C^\dagger
PC \rho C^\dagger P'C
\end{eqnarray}
This concludes the proof since the sum must vanish.

\end{proofof}

\begin{proofof}{\Le{Decompose}}
  We analyze the action of $U$ on the density matrix $\mcI\otimes\tau$. We first
  notice that since $U$ is a trace preserving operator, that is:
$\tr(U(\mcI\otimes\tau) U^\dagger\big)=\tr(\mcI\otimes\tau\big)=1$.
On the other hand it holds that:
\begin{eqnarray}
\tr\big(U(\mcI \otimes \tau) U^{\dagger}\big) &=&
\sum_{P,P'\in \mbP_n}\tr\big((P\otimes U_P)(\mcI\otimes\tau)
(P'\otimes U_{P'})^{\dagger}\big)\\
&=&
\sum_{P,P'\in \mbP_n}\tr\big(P\mcI P'^{\dagger} \otimes U_P\tau U_{P'}^{\dagger}\big)\\
&=&
\sum_{P,P'\in \mbP_n}\tr\big(PP'^{\dagger}\big) \tr\big( U_P\tau U_{P'}^{\dagger}\big)
\end{eqnarray}
If $P\ne P'$ then $\tr\big(PP'^{\dagger}\big)=0$, and therefore:
\begin{eqnarray}\begin{aligned}
\ldots &=
\sum_{P\in \mbP_n}\tr\big(\mcI\big) \tr\big( U_P\tau U_{P}^{\dagger}\big)\\
&=
\sum_{P\in \mbP_n} \tr\big( U_P\tau U_{P}^{\dagger}\big)
\end{aligned}\end{eqnarray}
It follows that  $1=\sum_{P\in \mbP_n} \tr( U_P\tau U_{P}^{\dagger})$, which concludes the proof.

\end{proofof}

\subsection{Concatenated Clifford \QAS}
\label{app:concatclifford}
\begin{proofof}{\Th{thm:CliffordConcat}}
  From \Cl{claim:EveAttack}, we know that any attack by Eve on an authenticated
  register is equivalent to an effect of the mixing operator $\mcM_s$, on the unencoded message space. We notice that any attack on
  the concatenated protocol is in fact equivalent to separate attacks on the
  different registers. This fact follows from the fact any individual attack can
  be broken down to attacks of the form $\mcM_s$, specifically for $r=2$:
\begin{equation}
\begin{split}
\rho_{Bob}= &\frac 1 {|\mfC_n|^2}\sum_{C_1,C_2 \in \mfC_n}(C_1\otimes C_2)^\dagger E\left(
    (C_1\otimes
    C_2)(\rho_1\otimes\rho_2)(C_1\otimes C_2)^\dagger\right)(C_1\otimes C_2) \\
  =& \frac 1 {|\mfC_n|^2}\sum_{P,Q\in \mbP_n} \sum_{C_1,C_2 \in
    \mfC_n}\alpha_{P,Q}(P\otimes Q) (C_1\otimes
  C_2)(\rho_1\otimes\rho_2)(C_1\otimes C_2)^\dagger)(P\otimes Q)^\dagger \\
  =&\sum_{P,Q\in \mbP_n}\alpha_{P,Q} \big(\frac 1 {|\mfC_n|}\sum_{C_1\in
    \mfC_n}(C_1^\dagger PC_1\rho_1C_1^\dagger P^\dagger C_1)\big)\otimes
  \big(\frac 1 {|\mfC_n|}\sum_{C_2\in \mfC_n}(C_2^\dagger QC_2\rho_2C_2^\dagger
  Q^\dagger C_2)\big)
\end{split}
\end{equation}
We denote $\mcM_s(\rho)=\wt\rho$, and use \Le{mix}: it holds that:
\begin{align}
\label{eq:concat}
\ldots=\alpha_{\mcI,\mcI}(\rho_1\otimes \rho_2) +
(\sum_{P,Q\ne \mcI \in \mbP_n}\alpha_{P,Q})( \wt\rho_1\otimes  \wt\rho_2) +
(\sum_{P\ne \mcI \in \mbP_n}\alpha_{P,\mcI})( \wt\rho_1\otimes  \rho_2) +
(\sum_{Q\ne \mcI \in \mbP_n}\alpha_{\mcI,Q})( \rho_1 \otimes  \wt\rho_2)
\end{align}
Bob does not abort, if both individual Clifford \QAS\ are valid. From the
security of the individual \QAS\ we know that
$\tr{((P_0^{\rho_i})B(\wt\rho_i))}<2^{-d}$ where $B$ is Bob cheat detecting
procedure. We also notice that $P_0^{\rho_1\otimes\rho_2}=P_0^{\rho_1}\otimes
P_0^{\rho_2}$.  We first rewrite \Eq{eq:concat} more clearly:
\begin{eqnarray} {\ldots} =s(\rho_1\otimes\rho_2)
  +h(\wt\rho_1\otimes\rho_2)+r(\rho_1\otimes\wt\rho_2)+t(\wt\rho_1\otimes\wt\rho_2)
\end{eqnarray}
and using the above observations we have:
\begin{equation}
\begin{split}
  \tr\left(P_0^{\rho_1\otimes\rho_2} B\left(s(\rho_1\otimes\rho_2)
      +h(\wt\rho_1\otimes\rho_2)+r(\rho_1\otimes\wt\rho_2)+t(\wt\rho_1\otimes\wt\rho_2)\right)\right)\
  =\ & s\cdot0 +q\cdot2^{-d}+r\cdot2^{-d}+t\cdot2^{-d}\cdot2^{-d}\\
  \le&(1-a)2^{-d}
\end{split}
\end{equation}
Where the inequality holds since $s+q+r+t=1$.

\noindent The claim for $r>2$ is follows the exact same lines and therefore is omitted.

\end{proofof}
\section{Polynomial Authentication Scheme}
\subsection{Security Proof of Polynomial \QAS} \label{app:securitypolynomial}
\subsubsection{Security Against Pauli Attacks}
\begin{lem} \label{lem:PauliPolyAuthSec} The polynomial \QAS\ is secure against
(generalized) Pauli attacks, that is, in the case where the
adversary applies a Pauli operator. In this case the projection of
Bob's state on the space spanned by $P_1\ket{\psi}$ is at least
$1-2^{-d}$.
\end{lem}

\begin{proof}
  Let us consider the effect of a Pauli $Q$ operator on the signed polynomial
  code $\mcC_k$. We first show that with probability $1-2^{-d}$ over the sign
  key $k$, the effect of $Q$ is detected by the error detection procedure.

Let $Q_x\ne \mcI$ be a Pauli operator $Q_x=X^{x_i}\odots X^{x_m}$ where $x\in
F_q^m$. The effect of $Q_x$ on the code is an addition of $x_i$ to the $i'th$
qubit.  This addition passes the error detection step only if coincides with the
values of a \textbf{signed} polynomial of degree at most $d$. We consider two
cases depending on the weight of $x$:
\begin{itemize}
\item If $|x|\le d$: let us denote by $g$ the polynomial that satisfies
      $\forall_{i} k_ig(\alpha_i) = x_i$, since $Q_x\ne \mcI$ we know that $g\ne
      0$. then $g$ has at least $m-d$ zeros. Since $g$ is nonzero it must have
      degree at of least: $m-d=d+1$. Such an attack will be detected with
      certainty by the error detection procedure.
\item Otherwise, assume without loss of generality that $x_i \ne 0$ for
      $i\le|x|$. There is exactly one polynomial $f$ of degree at most $d$ such
      that $\forall_{i\ge d+1}\ k_if(\alpha_i)=x_i$. For the attack of Eve to be
      undetected $x$ must agree with $f$ on the remaining coordinates as well:
\begin{eqnarray}
\Pr_k(\forall_{i\le d}\ x_i =k_if(\alpha_i))=\prod_{i=1}^d\Pr_k(k_i = x_i^{-1}f(\alpha_i))
    \end{eqnarray}
    Equality holds since: $k_i$ are independent, $k_i=k_i^{-1}$ and $x_i\ne0$
    for $i\le d$. Since $k_i=c$ with probability at most half we conclude that
    the probability that Eve's attack is undetected is at most $2^{-d}$.
\end{itemize}
Now that we have proved the claim for operators of the form $Q_x$, we handle the
general case. Pauli $Z$ are mapped in the dual code to $X$ operators. Since the
signed polynomial code is self dual, $Q_z$ attacks will be caught with
probability $1-2^{-d}$ as well. To conclude the proof we notice that detection
$Q_x$ attacks do not depend on the existence $Q_z$ attacks, therefore, a non
identity operator $Q_{x,z}=P_zP_x$ will be detected with the correct probability
since either $x$ or $z$ must be non trivial.

What remains is to notice that the Pauli randomization $P_{x,z}$ simply shifts
any attack $Q$ on the authenticated message to a different Pauli. That is the
effect on the signed polynomial code is $P_{x,z}^\dagger QP_{x,z}$. We conclude
that any Pauli operator acting on the polynomial \QAS\ is detected with a
probability of at least $1-2^{-d}$ as claimed.
\end{proof}

\subsubsection{Security Against General Attacks}
We start with a generalization of \Le{clifTw} for generalized Pauli
operators.

\begin{lem} \label{pauliTw} Let $P \ne P'$ generalized Pauli operators. Then:
$\sum_{Q\in \mbP_m} Q^\dagger P Q \rho Q^\dagger P'^\dagger Q = 0$
\end{lem}

The proof follows the same line as \Le{clifTw}:

\begin{proofof}{\Le{pauliTw}}
Let  $P\ne P'$ be generalized Pauli operator $P=X^aZ^b$ and $P'=X^{a'}Z^{b'}$.
\begin{eqnarray}
\sum_{Q\in \mbP_m} Q^\dagger P Q \rho
Q^\dagger P'^\dagger Q
&=&
\sum_{d,c=0}^{q-1} (X^cZ^d)^\dagger X^aZ^b (X^cZ^d) \rho
(X^cZ^d)^\dagger (X^{a'}Z^{b'})^\dagger (X^cZ^d)
\end{eqnarray}
We use the fact that $Z^dX^c=\omega_q^{dc} X^cZ^d$ and some algebra:
\begin{eqnarray}
{\ldots} &=&
\sum_{d,c=0}^{q-1} \omega_q^{d(a-a')+c(b-b')} X^aZ^b \rho
Z^{-b'}X^{-a'}\\
&=& X^aZ^b \rho
Z^{-b'}X^{-a'} \sum_{c=0}^{q-1} \omega_q^{c(b-b')} \sum_{d=0}^{q-1}
\omega_q^{d(a-a')}
\end{eqnarray}
To conclude the proof we recall that $a\ne a'$ or $b\ne b'$, hence
one of the above sums vanishes.

\end{proofof}

In addition we need one more simple lemma:
\begin{lem} \label{pauliID}For any generalized Pauli operator $P$
\begin{eqnarray}
\frac 1 {\left|\mbP_m\right|}\sum_{Q\in \mbP_m} Q^\dagger P Q \rho
Q^\dagger P^\dagger Q = P\rho P^\dagger
\end{eqnarray}
\end{lem}

\begin{proofof}{\Le{pauliID}}
  From the observation about generalized Pauli operators in \Sec{Back} we know
  that for any two generalized Pauli operators $P,Q$ $PQ=\alpha QP$ where
  $\alpha$ is some phase dependent on $P$ and $Q$.
\begin{eqnarray}\begin{aligned}
\frac 1 {\left|\mbP_m\right|}\sum_{Q\in \mbP_m} Q^\dagger P Q \rho
Q^\dagger P^\dagger Q &=  \frac 1 {\left|\mbP_m\right|}\sum_{Q\in
\mbP_m} Q^\dagger(\alpha Q P)  \rho (\alpha^* P^\dagger Q^\dagger) Q
\\&=  \frac 1 {\left|\mbP_m\right|}\sum_{Q\in
\mbP_m} \alpha P\rho \alpha^* P^\dagger
\\&= P\rho  P^\dagger
\end{aligned}\end{eqnarray}
\end{proofof}

\begin{proofof} {\Th{thm:PolynomialAuth}} The proof will follow the same lines
  as \Th{thm:CliffordAuth}.  For clarity, we omit the normalization factor
  $|\mbP_m|$. We start by decomposing any attack \hbox{$V\in \mathbbm{U}(M\otimes E)$}
  made by Eve to \mbox{$V=\sum_{P\in \mbP_m}P\otimes U_P$}. Bob's state prior to
  applying the error detection procedure is:
\begin{eqnarray}
\rho_{Bob} &=&  \tr_E{\Big(\sum_{Q\in\mbP_m}(Q\otimes \mcI_E)^\dagger
V\left( (Q\otimes \mcI_E)\rho \otimes\rho_E(Q\otimes \mcI_E)^\dagger
\right)V^\dagger(Q\otimes \mcI_E)\Big)} \\
&=&
\tr_E{\Big(\sum_{P,P'\in\mbP_m}\sum_{Q\in\mbP_m}(Q\otimes \mcI_E)^\dagger
P\otimes U_P\left( (Q\otimes \mcI_E)\rho\otimes\rho_E
(Q\otimes \mcI_E)^\dagger \right) P'\otimes U_{P'}^\dagger(Q\otimes \mcI_E)\Big)}
\end{eqnarray}
Regrouping elements operating on $M$ and on $E$ we have:
\begin{eqnarray}\begin{aligned}
{\ldots} =&
\tr_E{\Big(\sum_{P,P'\in\mbP_m}\sum_{Q\in\mbP_m}\left(Q^\dagger
P Q\rho Q^\dagger
P'Q\right)\otimes  U_P\rho_{E} U_{P'}^\dagger \Big)}\\
=&
\sum_{P,P'\in\mbP_m}
\sum_{Q\in\mbP_m}\tr{\left(U_P\rho_{E}U_{P'}^\dagger \right)} \cdot
\left(Q^\dagger P Q\rho Q^\dagger  P'Q\right)
\end{aligned}\end{eqnarray}
We use \Le{pauliTw} and are left only with $P=P'$
\begin{eqnarray}\label{prev2}
{\ldots}  =
\sum_{P\in\mbP_m}\sum_{Q\in\mbP_m}
\tr{\left(  U_P\rho_{E}U_P^\dagger  \right)} \cdot\left(Q^\dagger
P Q\rho Q^\dagger
PQ\right)
\end{eqnarray}
Now we use \Le{pauliID} :
\begin{eqnarray}
{\ldots}  =
\sum_{P\in\mbP_m}
\tr{\left( U_P^\dagger U_P\rho_{E} \right)}\cdot|\mbP_m
|P \rho  P^\dagger
\end{eqnarray}
We set $\alpha_P = \tr{\big( U_P^\dagger U_P\rho_{E} \big)} $ and
we rewrite Bob's state after normalization:
\begin{eqnarray}
\alpha_{\mcI}\cdot \rho + \sum_{P\in\mbP_m\setminus\{ \mcI\}}
\alpha_P\cdot P
\rho P^\dagger
\end{eqnarray}\
Recall that we are interested projection of Bob's state on the
subspace spanned by the operator $P_1^{\ket\psi}$.
\begin{eqnarray}
\tr\Big( P_1^{\ket\psi} \big(\alpha_{\mcI}\cdot \rho + \sum_{P\in\mbP_m\setminus \{\mcI\}}
\alpha_P\cdot P\rho P^\dagger  \big)  \Big) &=&
\alpha_\mcI
+\sum_{P\in\mbP_m\setminus \{\mcI\}}\alpha_P\tr\Big(P_1^{\ket\psi}
P\rho P^\dagger\Big)
\end{eqnarray}
We use the bound from \Le{lem:PauliPolyAuthSec}:
\begin{eqnarray}\label{eq:FinalCalc}
{\ldots} &\ge& \alpha_\mcI
+\sum_{P\in\mbP_m\setminus \{\mcI\}}\alpha_P
\left(1-2^{-d}\right) \\
&=& (1-\frac {1-\alpha_\mcI} {2^d})
\end{eqnarray}
Which concludes the proof. Similarly to the random Clifford
authentication scheme, the further Eve's intervention is closer to
the identity, that is -- Eve does almost nothing, then the projection
on the good subspace is closer to $1$.
\end{proofof}

\subsection{Concatenated Polynomial \QAS\ }\label{app:concatpolynomial}

When authenticating multiple registers, it may seem at first glance
that Eve has the advantage of being able to tamper with the state by
applying some transformation on the entire space. In the concatenated
Clifford authentication protocol, the intervention of Eve is broken
down to individual attacks on each register by the fact random
Clifford operators are applied to each register independently.

The main idea for the concatenated polynomial authentication is to use
an independent Pauli key $(x,z)$ for each registers, while maintaining
the sign key $k$ equal between registers. This idea will suffice to
``brake up'' the attack of Eve to a sequence of attacks on each
register separately.

\begin{protocol}\textbf{Concatenated polynomial Authentication protocol}:\\
Alice wishes to send a state $\ket\psi \in (\mcC^q)^{\otimes r}$
that is $r$ $q$-dimensional systems. For a security parameter $d$,
set $m=2d+1$. Alice randomly selects a single sign key $k\in \{\pm
1\}^m$, furthermore, Alice selects $r$ independent Pauli keys
$(x_j,z_j)$.

To encode $\ket\psi$ Alice encodes each $q$-dimensional system using
the signed polynomial code specified by $k$. Additionally, Alice
shifts the $j$'th encoded message by  $P_{(x_j,y_j)}$.

Bob decodes each message separately, if all messages are correctly
authenticated Bob dealers as valid the concatenated message, otherwise Bob
aborts.
\label{PolynomialConcatSec}
\end{protocol}

We now prove \Th{thm:concatpolynomial}.

\begin{proofof} (of \Th{thm:concatpolynomial})
  We notice that all the reasoning in \Th{thm:PolynomialAuth} till
  \Eq{eq:FinalCalc} hold in this case as well. So we have that the projection on
  the good subspace $P_1^{\ket\psi}$ is equal to:
\begin{eqnarray} \label{reFinalCalc}
\alpha_\mcI
+\sum_{P\in\mbP_{r\cdot m}\setminus \{\mcI\}}\alpha_P\tr\left(P_1^{\ket\psi}
P\rho P^\dagger\right)
\end{eqnarray}
We start by writing $\tr(P_1^{\ket\psi} P\rho P^\dagger) =
1-\tr(P_0^{\ket\psi} P\rho P^\dagger)$. We recall that $P$ here is a
Pauli operator from the group $\mbP_{m\cdot r}$ so we write:
$P=P_{\left(1\right)}\odots P_{(r)}$.
\begin{lem}
The probability for Bob to be fooled by the application of
$P\ne\mcI$ is at most $2^{-d}$.
\end{lem}
\begin{proof} For $P\rho P^\dagger$ to be in $P_0^{\ket\phi}$ it must be the
  case that for all $j$ such that $P_{(j)}\ne \mcI$ Eve escapes detection (Bob
  does not abort although the register is ``corrupted''). We note that Bob
  declares as valid the remaining registers (where $P_{(j)}= \mcI$) with
  certainty.  We assume without loss of generality that $P_{(1)}\ne \mcI$, we
  write the probability that Bob is fooled:
\begin{eqnarray}
\Pr{(\textit{Bob is fooled by } P)}
&=&\Pr{(\forall_{j: P_{(j)}\ne \mcI}\textit{Bob is fooled by } P_{(j)})}\\
&\le& \Pr{ (\textit{Bob is fooled by } P_{(1)}) }\\
&\le&  2^{-d}
\end{eqnarray}
Where the last inequality holds by \Le{lem:PauliPolyAuthSec}.
\end{proof}

Plugging this result into \Eq{eq:FinalCalc} we have:
\begin{eqnarray}
{\ldots}&=&
\alpha_\mcI +\sum_{P\in\mbP_{r\cdot m}\setminus \{\mcI\}}\alpha_P\left(1-\tr(P_0^{\ket\psi} P\rho P^\dagger)\right)\\
&\ge& \alpha_\mcI +\sum_{P\in\mbP_{r\cdot m}\setminus \{\mcI\}}\alpha_P\left(1-2^{-d}\right)\\
&=& \left(1-\frac {1-\alpha_\mcI} {2^d}\right)
\end{eqnarray}
Which concludes the proof.
\end{proofof}

\section{Polynomial Authentication Based \QPIP}
\subsection{Secure Application of Quantum Gates}
\label{app:polynomialgates}
We have seen in \Sec{sec:SignedPolynomial} how to perform operations on states
encoded by a polynomial code. In this section we present a way for the
prover to apply certain operations on a signed shifted Polynomial
error correcting code. This can be done without compromising the
security of the authentication scheme.

The main idea is that the transitive operation performed on the signed
Polynomial code have almost the desired effect on the state at hand. The
verifier will only need to update \textbf{his} keys $(x,z)$ for
the provers action to have the desired effect on the state.

We will first show the simple and elegant fact that if the verifier wants a
(generalized) Pauli applied to the state, he does not need to ask the prover to
do anything. The only thing the verifier must do is change his Pauli keys. Then,
we show how to perform other operations such as \textit{SUM}, Fourier and
Measurement.

\begin{itemize}
\item \textbf{Pauli $X$}: The logical $\wt{X}$ operator consists of an
      application of $X^{k_1}\otimes, {\ldots} \otimes X^{k_m}$ where
      \textbf{k} is the sign key. We claim that the change $(x,z)\rightarrow (x
      - \mathbf k,z)$ will in fact change the interpretation the
      verifier assigns to the state in the desired way.
\begin{equation}\begin{aligned}
P_{x,z}\ket{S^k_a} &= P_{x-k,z}   P_{x-k,z}^\dagger  P_{x,z}\ket{S^k_a}\\
&=  P_{x-k,z} X^{-(x-k)}Z^{-z} Z^zX^x\ket{S^k_a}\\
&=P_{x-k,z}(X^{k_1}\otimes, {\ldots} \otimes X^{k_m})\ket{S^k_a} \\
&=  P_{x-k,z}\wt{X} \ket{S^k_a}
\end{aligned}
\end{equation}

\item \textbf{Pauli $Z$}: Similarly to the $X$ operator, all that is needed is a
      change of the Pauli key. We recall that $\wt{Z}=Z^{r_1k_1}\odots
      Z^{r_mk_m}$. We define the vector $\mathbf{t}$ to be $t_i=c_ik_i$. From
      the same argument as above, it holds that the change of keys must be
      $(x,z) \rightarrow (x,z -\mathbf{t})$.

\item \textbf{Controlled-Sum}: In order to remotely apply the
      \textit{SUM} operation the prover perform transversely
      Controlled-Sum (\textit{SUM}) from register $A$ to register $B$
      on the authenticated states; as if the code was not shifted by
      the Pauli masking. However, a change in the Pauli keys is needed
      for the operation to have the desired effect. It is easy to
      check that:
\begin{eqnarray} {\textit{SUM}}
  (Z^{z_A}X^{x_A}\otimes Z^{z_B}X^{x_B}) &=&
  (Z^{z_A-z_B}X^{x_A}\otimes Z^{z_B}X^{x_B+x_A})
  {\textit{SUM}}
\end{eqnarray}
Which implies that the same hold for the logical operation
$\wt{\textit{SUM}}$, and the Pauli shift $P_{(x,z)}$ that is:
\begin{eqnarray} \wt{\textit{SUM}}\left(P_{(x_A,z_A) } \otimes
    P_{(x_B,z_B)}\right) &=&\left(P_{(x_A,z_A-z_B) }\otimes
    P_{(x_B+x_A,z_B)}\right) \wt{\textit{SUM}}
\end{eqnarray}
Hence, the verifier must change the pair of keys $(x_A,z_A), (x_B,z_B)$ to
$(x_A,z_A-z_B)$ and $ (x_B+x_A,z_B)$, for the \textit{SUM} to have the desired
affect on the state.

\item \textbf{Fourier}: The prover performs Fourier transversely on the
      authenticated state. We recall that the Fourier operation swaps the roles
      of the $X$ and $Z$ Pauli operator. $FX^xF^\dagger=Z^x$ and
      $FZ^zF^\dagger=X^{-z}$. This is true for each register separately and
      hence:
\begin{eqnarray}\begin{aligned} \wt{\textit{F}}\cdot
    Z^{z_1}X^{x_1}\otimes{\ldots}\otimes Z^{z_m}X^{x_m} =&
    X^{-z_1}Z^{x_1}\otimes{\ldots}\otimes X^{-z_m}Z^{x_m}\cdot
    \wt{\textit{F}}\\
    \simeq& Z^{x_1}X^{-z_1}\otimes{\ldots}\otimes
    Z^{x_m}X^{-z_m}\cdot\wt{\textit{F}}
\end{aligned}
\end{eqnarray}
Where the last equality is up to a global phase.

Therefore the verifier must  change the key $(x,z)$ to $(-z,x)$.

\item \textbf{Measurement in the standard basis}: The prover measures the
      encoded state in the standard basis, send result to the verifier. Using
      the $x$ part of Pauli key, and the knowledge of
      $k$, the verifier interpolates the polynomial according to values
      in the received set of points. If the polynomial is indeed a polynomial of
      low degree (which is always the case if the prover is honest) the verifier
      sends the encoded value to the prover. Otherwise, the prover is caught
      cheating and the verifier aborts.
\item \textbf{Toffoli}: The (generalized) Toffoli gate is applied using Clifford group
     operations on the Toffoli state $\frac 1 q \sum_{a,b} \ket{a,b,a b}$
     (\cite{benor2006smq,nielsen2000qcq}). The application of
      a Toffoli gate in such a way does not imply a change of keys
      directly. Changes are made with respect to the actual operations that were
      performed.
\end{itemize}
\subsection{Proof of \Le{lem:uniformkeys}}\label{apen:polyIPproof}
\begin{statement}{\Le{lem:uniformkeys}}
At any stage of the protocol the verifier's set of keys,
$k$ and $\{(x,z)_i\}_1^n$ are distributed uniformly and
independently.
\end{statement}

\begin{proofof}{\Le{lem:uniformkeys}} Before any gate is applied the claim
  holds. All that needs to be done it to check that all changes keep this
  desired property.

  The sign key $k$ does not change during the protocol so in this case the claim
  is trivial. At every step at most two pair of Pauli keys change, let us review
  the possible changes (see \pen{app:polynomialgates}) and verify that the
  claim holds:
\begin{itemize}
\item Changes from the Pauli operators and Fourier transform induces shift, swap or negation changes to the keys; all of them preserve the uniform independent distribution trivially.
\item The \textit{SUM} operation involves two set of keys $(x_A,z_A),(x_B,z_B)$ which change to $(x_A,z_A-z_B)$ and $ (x_B+x_A,z_B)$. The sum $x_B+x_A$, is $\mod q$ hence it is distributed uniformly, in addition it is not hard to see that is independent of $x_A$. The same holds for $z_A-z_B$ and $z_B$. Other parts of both keys are trivially distributed correctly.
\item When the prover measures in the standard basis an authenticated qubit the outcome of the measurement is distributed uniformly at random in $F_q^m$. Specifically, the outcome does not depend on the sign key or the information that is authenticated. Therefore, even when the prover has the interpretation of his measurement outcome, he does not gain any information about the sign key $k$ or the Pauli keys of other registers.
\end{itemize}
\end{proofof}

\section{Fault Tolerant \QPIP}\label{app:ft}

For the interactive proofs described above to be relevant in a any realistic
scenario, dealing with noise is necessary. We will present a
scheme based on the
polynomial \QPIP, that enables us to conduct interactive proofs in the presence
of noise.
\\~\\
\begin{statement}{\Th{thm:ft}}
{\it \Th{thm:qcircuit} holds also when the quantum communication
and computation devices are subjected to the usual local noise model}.
\end{statement}

\begin{proof}\textit{(Sketch)}
Our proof is based on a collection of
 standard fault-tolerant quantum computation techniques.
Care must be given to the fact that the verifier is the
only one who can authenticate qubits, while he cannot
authenticate many qubits in parallel.

The proof can be divided into three stages.

In the first stage, the prover receives authenticated qudits from the
verifier, one by one. Each qudit is authenticated on $m$ qudits.
The prover ignores the authentication structure and begins encoding
each qubit out of the $m$ qubits
separately using a concatenated error correction code,
with total length which is polylogarithmic, as is
required for the fault tolerance scheme in \cite{aharonov1997ftq}.
From the work
of \cite{aharonov1997ftq,knill1998rqc} (and others) we know that this encoding
can be done in a fault tolerant way, such that if the error
probability was less
than some threshold $\eta$, then the encoded qudit is faulty (namely,
has an effective error) with
probability at most $\eta'$, where $\eta'$ is
a constant that depends on $\eta$ and other parameters of the
encoding scheme, but not on $n$.
We denote this concatenated encoding procedure by $\bar S$.

 Since each
authenticated qudit sent to the prover
is encoded using a constant number ($m$) of qudits, it
follows that with a constant probability,
$\eta''$ all these qubits are effectively correctly authenticated.
In other words, the encoding of
$\ket{S^k_a}$, $(\bar S \odots\bar S) \ket{S^k_a}$, has no effective
faults with probability $\eta''$.

Once a qudit has been encoded by the prover, he can keep applying
error corrections on that qudit, and thus, can keep its effective
error below some constant for a polynomially long time.
Polynomially many authenticated qudits are sent this way to the prover.

In the second stage a purification procedure is performed on the authenticated
messages, which are now protected from noise by the prover's concatenated error
correction code. Since the purification is of the \emph{authenticated} qudits,
it is done according to instructions from the verifier. As explained in
\pen{pen:polynomialgates} the verifier can also interpret measurements
outcomes for the prover, which are needed for the purification procedure. We
need to purify both input qubits which are without loss of generality $\ket 0$,
and Toffoli states. Any standard purification procedure (for example, that
of \cite{benor2006smq}) would work for the $\ket 0$
states. In order to purify the Toffoli states we
use the purification described in \cite{benor2006smq}. The purification
procedure uses polylogarithmically many qubits in order
to provide a total error
of at most $\frac \Delta {poly(nT)}$, where $T$ is the number of gates
in the circuit $U$ that will be computed by the prover.
This means (using the union bound) that with probability at most
$\Delta$ all purified states are effectively correct.

Finally, having with high probability, correct input states, the polynomial
\QPIP\ (\Prot{prot:PolynomialIP}) is executed. The prover applies logical
operations (\textit{SUM},$F$ and measurements) on his registers which contain
authenticated qubits. In particular, a logical measurement of the output bit of
the computation is executed by the prover at the end of the computation. The
result is then sent to the verifier who subsequently
interprets it according to his secret key.

The soundness of the this fault tolerant \QPIP\ is the same as that of
the standard \QPIP. In fact, in this scheme, the verifier ignores the prover's
overhead of encoding the input in an error correcting code, and performing
encoded operations. The verifier can be thought as performing
\Prot{prot:PolynomialIP} for a purification circuit followed by the circuit he
is interested in computing.  Therefore, the security proof of
\Th{thm:PolynomialIP} proves in fact that applying the purification and
computation circuits, has the same soundness parameter as the standard \QPIP.

Regarding completeness, the fact that the prover's computation
is noisy changes the error probability only very slightly.
There is a probability $\Delta$ that one of the input authenticated
states is effectively incorrect; Once they are all correct,
the fault tolerance proof implies that they remain correct
the entire computation with all but an inverse polynomial
probability.
Therefore, if the standard
\QPIP\ protocol has completeness $1-\delta-\epsilon$ the
completeness of this scheme is bounded by $1-\delta-\epsilon-2\Delta$.
\end{proof}

\section{Blind \QPIP}\label{app:blind}
\begin{deff} \cite{blind,broadbent2008ubq,childs2001saq}
  Secure blind quantum computation is a
  process where a server computes a function for a client and the following
  properties hold:
\begin{itemize}
\item \textbf{Blindness}: The prover gets no information beyond an upper bound on the size of the circuit.
    Formally, in a blind computation scheme for a set of function  $\mathfrak{F}_n$
      the prover's reduced density matrix is identical for every
      $f\in\mathfrak{F}_n$.
\item \textbf{Security}: Completeness and soundness hold the same way as in
      \QAS\ (\Def{def:qas}).
\end{itemize}

\end{deff}

\begin{statement}{\Th{thm:blind}} There is a blind \QPIP\ for \qc.
\end{statement}

We use the \QPIP\ protocols for \qc\ in order to provide a blind protocol for
any language in \BQP. We use the simple observation that the input is completely
hidden from the prover. This holds since in both \QAS\ presented the density
matrix that describes the prover's state does not depend on the input to the
circuit. Specifically, due to the randomized selection of an authentication, the
prover's state is the completely mixed state. We also use in the proof of this
theorem the notion of a universal circuit. Roughly, a universal circuit acts on
input bits and control bits. The control bits can be thought of, as a description of
a circuit that should be applied to the input bits. Constructions of such
universal circuits are left as an easy exercise to the reader.

Having mentioned the above observations, a blind computation protocol is not
hard to devise. The verifier will, regardless of the input, compute, with the
prover's help, the result of the universal circuit acting on input and control
bits.

We first formally define a universal circuit:
\begin{deff} The universal circuit $\mathfrak{U}_{n,k}$ acts in the
  following way:
  \begin{eqnarray}
    \mathfrak{U}_{n,k} \ket\phi\otimes\ket{c(U)}
    \longrightarrow U\ket\phi\ket{c(U)}
  \end{eqnarray}
  Where $c(U)$ is the canonical (classical) description of the circuit $U$.
\end{deff}

\begin{proofof}{\Th{thm:blind}}
  We prove that both the Clifford based \QPIP\ and the Polynomial \QPIP\ can be
  used to create a blind computation protocol. We claim that the state of the
  prover through the protocols is described by the completely mixed state. This
  is true in the Polynomial scheme since the Pauli randomization does exactly
  that. Averaging over all possible Pauli keys, it is easy to check that the
  state of the prover is described by $\mcI$. Furthermore, the prover gains no
  information regarding the Pauli key during the protocol, therefore, the
  description of the state does not change during the protocol as claimed.

  Since the above holds for any initial state, it follows that the prover has no
  information about the initial, intermediate or final state of the system.

  To see that the same argument holds for the Clifford \QAS, it suffices to
  notice that applying a random Clifford operator ``includes'' the application
  of a random Pauli:
\begin{eqnarray}
  \frac 1 {|\mfC_n|}\sum_{c\in \mfC_n} C\rho C^\dagger &=& \frac 1 {|\mfC_n|}\sum_{c\in \mfC_n} (CQ)\rho (CQ)^\dagger
\end{eqnarray}
Equality holds for any $Q\in \mfC_n$ since it is nothing but a change of order
of summation.
\begin{eqnarray}
  \ldots  &=& \sum_{Q\in \mbP_n}\frac 1 {|\mbP_n|} \frac 1
  {|\mfC_n|}\sum_{c\in \mfC_n} C(Q\rho Q^\dagger)C^\dagger\\
  &=& \frac 1
  {|\mfC_n|}\sum_{c\in \mfC_n}\frac 1 {|\mbP_n|}\sum_{Q\in \mbP_n}
  C(Q\rho Q^\dagger)C^\dagger\\
  &=&\frac 1
  {|\mfC_n|}\sum_{c\in \mfC_n}C\Big(\frac 1 {|\mbP_n|}\sum_{Q\in \mbP_n}
  \left(Q\rho Q^\dagger\right)\Big)C^\dagger\\
  &=&\frac 1
  {|\mfC_n|}\sum_{c\in \mfC_n}C\left(\mcI\right)C^\dagger\\
  &=& \mcI
\end{eqnarray}
\end{proofof}

\ignore{
\section{\QPIP\ With Different Provers} \label{app:interpretation}
We have seen that a prover restricted to \BQP\ has enough power to convince a
verifier of membership in any language in \BQP. We would like to find out how
strict are the requirements on the prover's computational power for our results
to hold. We consider two extreme cases:

\begin{itemize}
\item A quantum-classical hybrid prover. This prover which we denote $\mcP_\mathbf{hybrid}$
      has, similarly to our verifier, constant amount of quantum memory on which
      he can act. In addition we assume no other (classical) computational bound
      on the prover.
\item A computationally unbounded quantum prover which we denote $\mcP_\mathbf{unbound}$. This
      prover is identical to other systems studied in the literature
      (for instance: \cite{watrous2003phc}).
\end{itemize}

We would like to know what are the properties of our \QPIP\ protocols, when
interacting with such provers. Namely, do the soundness and completeness proofs
hold?

We prove an interesting feature of our \QPIP\ protocols: A prover of type
$\mcP_\mathbf{hybrid}$ is unable to convince $\mcV$ of even a true statement.  That is, $\mcV$
will abort (with high probability) an interaction with such a $\mcP_\mathbf{hybrid}$.
\\~\\
\begin{statement}{\Cl{claim:hybridProver}} There exists languages in \BQP\ such
  that for sufficiently large instances $\mcV$ will abort with high probability.
\end{statement}

\begin{proof}\textit{(Sketch. The complete proof will appear in the full version.)}
  A $\mcP_\mathbf{hybrid}$ prover cannot store $n>c$ qubits. Assume that
  $\mcP_\mathbf{hybrid}$ receives, the authentication of a random string. Then
  similarly, to the encryption procedure of  \cite{bennett1984qcp}, the
  prover is asked to measure a random qubit either in the standard or Fourier
  basis. Since $\mcP_\mathbf{hybrid}$ could not have kept the needed qubit in a
  coherent state, it is possible to show that he will fail to provide an
  acceptable answer with probability of at least $ \frac 1 2 (1-o(1))(1-\frac 1
  {2^d})$. The proof is similar to the security proof of the
  \cite{bennett1984qcp}. Intuitively, For $\mcP_\mathbf{hybrid}$ to be caught he
  must not have saved the bid and have guessed wrongly whether the standard or
  Fourier basis will be requested. This happens with probability $(1-\frac c
  n)\frac 1 2$. Now the string that $\mcP_\mathbf{hybrid}$ holds is independent
  of that correct set of strings, so $\mcP_\mathbf{hybrid}$ must guess a bit and
  fake an authentication for it. This succeeds with probability of at most
  $2^{-d}$.
\end{proof}

As for the computationally unbounded quantum prover, it is quite obvious that
$\mcP_\mathbf{unbound}$ can be used in order to convince $\mcV$ of membership in a \BQP\
language using our protocols. However, it may come as a surprise to find out
that the extra power of $\mcP_\mathbf{unbound}$ does not enable him to cheat $\mcV$ more than
the regular prover $\mcP$.

\begin{thm}\label{thm:unboundProver} The completeness and soundness of the
  \QPIP\ protocols \ref{prot:CliffordIP} and \ref{prot:PolynomialIP} remains the
  same when interacting with a $\mcP_\mathbf{unbound}$ prover.
\end{thm}

\begin{proof}
  To see this, we notice that in the security proofs for the \QPIP\ protocols
  (and for the \QAS) we had no assumptions on the computational power of the
  prover (or Eve). That is, the security proofs holds for \textbf{any}
  intervention, regardless of the possible complexity it might encapsulate in
  it.
\end{proof}

Thus, the security of a delegated (possibly blind) quantum computation is
enhanced by \Th{thm:unboundProver}, since it abandons the assumption that the
prover is computationally bound to \BQP. That is, our security proofs are
computationally unconditional.}

\section{Interpretation of Results}\label{app:interpretation}
\begin{proofof}{Corollary \ref{corol:confidence}}
  Let us first deal with the Clifford based \QPIP. We assume that the soundness
  of the scheme is $\delta$ and that the prover applies a strategy on which the
  verifier dose not abort with probability $\gamma$.  The final state of the protocol before the verifier's cheat detection can be
  written as (see \Eq{attackProfile}):
\begin{equation}
s\rho_c  + \frac{(1-s)}{4^n-1}\sum_{Q\in\mbP_n\setminus\{\mcI\}}\left(Q\rho_c Q^\dagger
\right)
\end{equation}
Where $\rho_c$ is the correct final state of the protocol. After the verifier
applies the cheat detection procedure $\mcB$ (which checks that the control
registers are indeed in the $\ket 0$ state):

\begin{equation} s\rho_{c}\otimes\ket{VAL}\bra{VAL}
  +\alpha_{rej}\rho_{rej}\otimes\ket{ABR}\bra{ABR}
  +\alpha_{bad}\rho_{bad}\otimes\ket{VAL}\bra{VAL}
\end{equation}
Assume the verifier declares the computation valid, then his
state is:
\begin{equation}\begin{aligned}
    \frac {s\rho_{c}  +\alpha_{bad}\rho_{bad}} {1-\alpha_{rej}} \otimes\ket{VAL}\bra{VAL}
\end{aligned}\end{equation}
then the trace distance to the correct state $\rho_c$ is bounded by:
\begin{equation}
1 - \frac {s}{1-\alpha_{rej}}+ \frac {\alpha_{bad}}{1-\alpha_{rej}} = \frac {2\alpha_{bad}}{1-\alpha_{rej}}\le \frac{2\delta}{\gamma}
\end{equation}
Were the inequality follows from the security of the \QPIP\ protocol:
$\alpha_{bad}\le\delta$, and the fact that the non-aborting probability $\gamma$ is
equal to $\alpha_{bad}+s$.

The proof that the Polynomial based \QPIP\ has the same property follows the
exact same lines.
\end{proofof}

\section{Symmetric Definition of \QPIP}\label{app:sym}

The definitions and results presented so far seem to be asymmetric. They refer
to a setting where the provers wishes to convince the verifier \emph{solely} of
YES instances (of problems in \BQP). This asymmetry does not seem natural neither regarding the complexity class
\BQP\ nor in the cryptographic or commercial aspects. In fact, this intuition is
indeed true.

Apparently, we can provide a symmetric definition of quantum prover interactive
proofs, and show that the two definitions are equivalent.  Essentially, this
follows from the trivial observation that the class \BQP\ is closed under
complement, that is, $\mcL\in\BQP \iff {\mcL^c}\in\BQP$.

To see this, let us first consider the symmetric definition for \QPIP.

\begin{deff} A language $\mcL$ is in the class symmetric quantum prover
  interactive proof $(\QPIP^{sym})$ if there exists an interactive protocol with
  the following properties:
\begin{itemize}
\item The prover and verifier computational power is exactly the same as in the
      definition of \QPIP\ (\Def{def:QPIP}). Namely, a \BQP\ machine and
      quantum-classical hybrid machine for the prover and verifier respectively.
\item Communication is identical to the \QPIP\ definition.
\item The verifier has three possible outcomes: \textbf{YES},
      \textbf{NO}, and \textbf{ABORT}:
      \begin{itemize}
  \item\textbf{YES}: The verifier is convinced that $x\in \mcL$.
  \item\textbf{NO}: The verifier is convinced that $x\notin \mcL$.
  \item\textbf{ABORT}: The verifier caught the prover cheating.
      \end{itemize}

\item Completeness: There exists a prover $\mcP$ such that $\forall
      x\in\{0,1\}^*$ the verifier is correct with high probability:
\[ \Pr{(\left[\mcV,\mcP \right](x,r) = \mathbbm{1}_\mcL })  \ge \frac23 \]
where $\mathbbm{1}_\mcL$ is the indicator function of $\mcL$.
\item Soundness: For any prover $\mcP'$ and for {\bf any}
      $x\in\{0,1\}^*$, the verifier is mistaken with bounded
      probability, that is:
\[ \Pr{(\left[\mcV,\mcP \right](x,r) = 1- \mathbbm{1}_\mcL })  \le \frac13\]
\end{itemize}
\end{deff}

\begin{thm} For any language $\mcL$:  If $\mcL,\mcL^c$ are both in $\QPIP$ then $\mcL,\mcL^c \in \QPIP^{sym}$
\end{thm}

\begin{proof}
  Let $\mcV_{\mcL},\mcP_{\mcL}$ denote the \QPIP\ verifier and prover for the
  language $\mcL$. By the assumption, there exists such a pair for both $\mcL$
  and $\mcL^c$. We define the pair $\wt{\mcP}$ and $\wt{\mcV}$ to be
  $\QPIP^{sym}$ verifier and prover in the following way: On the first round the
  prover $\wt{\mcP}$ sends to $\wt\mcV$\; ``yes'' if $x\in \mcL$ and ``no''
  otherwise. Now, both $\wt{\mcP}$ and $\wt{\mcV}$ behave according to
  $\mcV_{\mcL},\mcP_{\mcL}$ if ``yes'' was sent or according to
  $\mcV_{\mcL^c},\mcP_{\mcL^c}$ otherwise. Soundness and completeness follows
  immediately from the definition.

\end{proof}

Since \BQP\ is closed under  completion, we get:
\begin{corol}
$\BQP=\QPIP^{sym}$
\end{corol}

\end{document}

%% file: def.tex
\usepackage[usenames]{color}
\usepackage{bbm,amssymb,amsmath}

\newcommand{\BQP}{{\sf BQP}}
\newcommand{\NP}{{\sf NP}}

\newcommand{\BPP}{{\sf BPP}}
\newcommand{\IP}{{\sf IP}}

\newcommand{\PSPACE}{{\sf PSPACE}}
\newcommand{\Eq}[1]{Eq.~\ref{#1}}
\newcommand{\Le}[1]{Lemma~\ref{#1}}
\newcommand{\Cl}[1]{Claim~\ref{#1}}
\newcommand{\Th}[1]{Theorem~\ref{#1}}
\newcommand{\Prot}[1]{Protocol~\ref{#1}}

\newcommand{\Def}[1]{Definition~\ref{#1}}

\newcommand{\Sec}[1]{Sec.~\ref{#1}}

\newcommand{\pen}[1]{Appendix~\ref{#1}}

\newcommand{\EqDef}{\stackrel{\mathrm{def}}{=}}

\newcommand{\bra}[1]{\left< #1\right|}
\newcommand{\ket}[1]{\left| #1\right>}

\newcommand{\tr}{\mbox{Tr}}

\newcommand{\mN}{\mathbbm{N}}
\newcommand{\mbP}{\mathbbm{P}}

\newcommand{\mcI}{\mathcal{I}}

\newcommand{\mcA}{\mathcal{A}}
\newcommand{\mcB}{\mathcal{B}}
\newcommand{\mcC}{\mathcal{C}}

\newcommand{\mcM}{\mathcal{M}}
\newcommand{\mcO}{\mathcal{O}}
\newcommand{\mcK}{\mathcal{K}}
\newcommand{\mcP}{\mathcal{P}}
\newcommand{\mcV}{\mathcal{V}}
\newcommand{\mcL}{\mathcal{L}}

\newcommand{\mfC}{\mathfrak{C}}

\newcommand{\ignore}[1]{}

\newtheorem{thm}{Theorem}[section]

\newtheorem{deff}{Definition}[section]

\newtheorem{protocol}{Protocol}[section]
\newtheorem{claim}[thm]{Claim}
\newtheorem{lem}[thm]{Lemma}

\newtheorem{corol}[thm]{Corollary}
\newtheorem{fact}[thm]{Fact}

\newenvironment{proofof}[1]{\vspace*{2mm}\noindent{\textit Proof} of $\bf{#1}$:\hspace*{1em}}{$\Box $}
\newenvironment{statement}[1]{\noindent{\textbf {#1}\hspace*{0.5em}}}{$\\$}

\def\hpic #1 #2 {\mbox{$\begin{array}[c]{l}
      \epsfig{file=#1,height=#2} \end{array}$}}

\def\vpic #1 #2 {\mbox{$\begin{array}[c]{l}
      \epsfig{file=#1,width=#2} \end{array}$}}